\documentclass[pra,twocolumn,floatfix,superscriptaddress,notitlepage,longbibliography]{revtex4-1}

\usepackage[normalem]{ulem}
\usepackage{graphicx, color, graphpap}% Include figure files
\usepackage{enumitem}
\usepackage{amssymb}
\usepackage{amsthm}
\usepackage{multirow}
\usepackage[dvipsnames]{xcolor}
\usepackage[colorlinks=true,urlcolor=MidnightBlue,citecolor=MidnightBlue,linkcolor=MidnightBlue]{hyperref}
\usepackage[T1]{fontenc}
\usepackage{verbatim}
\usepackage{mathtools}
\usepackage{titlesec}
\usepackage{float}

\long\def\ca#1\cb{} %Use for commenting out: \ca...\cb

\newcommand{\braket}[2]{\langle #1 \hspace{1pt} | \hspace{1pt} #2 \rangle}

\newcommand{\ket}[1]{|#1\rangle}               %ket
\newcommand{\bra}[1]{\langle #1|}              %bra

\newcommand{\CC}{\mathcal{C}}

\newcommand{\HC}{\mathcal{H}}

\newcommand{\PC}{\mathcal{P}}

\newcommand{\SC}{\mathcal{S}}

\newcommand{\ZC}{\mathcal{Z}}

\renewcommand{\geq}{\geqslant}
\renewcommand{\leq}{\leqslant}

\renewcommand{\vec}[1]{\boldsymbol{#1}}  % Bold vectors instead of arrow vectors

\usepackage{circuitikz}
\usepackage{tikz}
\usetikzlibrary{arrows, shapes.gates.logic.US, calc}
\usetikzlibrary{arrows,decorations.pathreplacing}
\usetikzlibrary{backgrounds,fit,decorations.pathreplacing,calc}

{}%         Body font
{}%         Indent amount (empty = no indent, \parindent 
\newtheorem{theorem}{Theorem}
\newtheorem{lemma}{Lemma}

\newtheorem{proposition}{Proposition}

\newtheorem{definition}{Definition}

\newtheorem{fact}{Fact}
\newcommand{\tr}[1]{\text{tr}\left[#1\right]}
\newcommand{\E}[1]{\mathbb{E}\left[#1\right]}

\begin{document}
\title{Multipartite entanglement measures via Bell basis measurements}

\author{Jacob L. Beckey}
\affiliation{JILA, NIST and University of Colorado, Boulder, Colorado 80309, USA}
\affiliation{Department of Physics, University of Colorado, Boulder, Colorado 80309, USA}

\author{Gerard Pelegr\'i}
\affiliation{Department of Physics and SUPA, University of Strathclyde, Glasgow G4 0NG, UK}

\author{Steph Foulds}
\affiliation{Physics Department, Durham University, South Road, Durham, DH1 3LE, UK}

\author{Natalie J. Pearson}
\affiliation{Department of Physics and SUPA, University of Strathclyde, Glasgow G4 0NG, UK}

\begin{abstract} 
We show how to estimate a broad class of multipartite entanglement measures from Bell basis measurement data. In addition to lowering the experimental requirements relative to previously known methods of estimating these measures, our proposed scheme also enables a simpler analysis of the number of measurement repetitions required to achieve an $\epsilon$-close approximation of the measures, which we provide for each. We focus our analysis on the recently introduced Concentratable Entanglements  [Beckey \textit{et al.} \href{https://doi.org/10.1103/PhysRevLett.127.140501}{Phys. Rev. Lett. 127, 140501 (2021)}] because many other well-known multipartite entanglement measures are recovered as special cases of this family of measures. We extend the definition of the Concentratable Entanglements to mixed states and show how to construct lower bounds on the mixed state Concentratable Entanglements that can also be estimated using only Bell basis measurement data. Finally, we demonstrate the feasibility of our methods by realistically simulating their implementation on a Rydberg atom quantum computer.
\end{abstract}
\maketitle

\textit{Introduction.}  
The precise control over quantum systems demonstrated in the past two decades has enabled rapid progress in the experimental study of quantum entanglement \cite{horodecki2009quantum,friis2018entanglement}. Entanglement plays an important role in enabling emerging quantum technologies to outperform their classical counterparts, with the degree and type of entanglement within the state determining its usefulness for a given task. Consequently the empirical characterization of entanglement is a problem of ubiquitous interest in quantum information science. While bipartite entanglement is well understood theoretically \cite{vedral2002uniqueness,horodecki2009quantum} and is routinely estimated in experimental settings, multipartite entanglement remains challenging to understand theoretically and probe experimentally \cite{friis2018entanglement}. When these considerations are coupled with the exponential scaling of the Hilbert space of multipartite systems, which makes quantum state tomography intractable at scale \cite{haah2017sample,odonnell2015efficient}, it is clear that there is a need for more experimentally efficient methods of multipartite entanglement quantification.

Recently, the authors of Ref. \cite{foulds2021controlled} conjectured that the output probabilities of the so-called \textit{parallelized c-SWAP test}, shown in Fig. \ref{fig:SWAPn}, could be used to construct a well-defined multipartite entanglement measure. The authors of Ref. \cite{beckey2021computable} then generalized this conjecture and proved that a whole family of multipartite entanglement measures could be constructed using the output probabilities of this circuit, depending on which ancilla qubits are measured. The resultant family of measures was dubbed the Concentratable Entanglements (CEs), and it was shown that many well-known multipartite entanglement measures could be recovered as special cases of this general family. Since their introduction, several interesting properties and applications of the CEs have also been studied \cite{cullen2022calculating,schatzki2021entangled,schatzki2022hierarchy}. We also note that the $n$-tangle \cite{wong2001potential}, another well-studied entanglement monotone, can be estimated via the parallelized c-SWAP test \cite{beckey2021computable}, and that the parallelized c-SWAP test was recently generalized to qudit and optical states \cite{prove22}.

From Fig.~\ref{fig:SWAPn}(a), it is  clear that the $n$-qubit c-SWAP test requires $n$ Toffoli gates as well $3n$ qubits (2 copies of the the quantum state of interest and $n$ ancilla qubits). The most promising platform for implementing the c-SWAP test is Rydberg atom systems \cite{saffman2010quantum,adams2019rydberg} due to their native ability to implement Toffoli gates \cite{brion07a,isenhower11,shi18,su18,su18a,beterov18,Levine2019parallel,khazali20,rasmussen20,li21,young21,pelegri2022high}. However, to make the CEs and related measures as accessible as possible, a method of estimating them that is experimentally feasible on all hardware platforms is needed. This work addresses this problem by introducing a method of estimating many multipartite entanglement measures from Bell basis measurement data -- an ancilla-free scheme that only requires one- and two-qubit gates acting on two copies of the quantum state of interest.

Bell basis measurements have played a crucial role in quantum information theory since the advent of protocols like quantum teleportation and superdense coding \cite{bennett1992communication,bennett1993teleporting,nielsen2000quantum}. More recently, Bell basis measurements have been implemented experimentally to estimate bipartite concurrences \cite{walborn2006experimental,walborn2007experimental}, non-stabilizerness (i.e. magic) \cite{haug2022scalable}, entanglement dynamics in many-body quantum systems \cite{daley2012measuring, islam2015measuring, kaufman2016quantum, bluvstein2022quantum}, and even to demonstrate quantum advantage in learning from experiments \cite{huang2022quantum}. These recent experiments corroborate the claim that our methods are feasible on today's hardware.

A limitation recently highlighted in Ref. \cite{cullen2022calculating} is that CEs were only well-defined on pure states. We address this limitation by first defining the CEs for mixed state inputs and then introducing lower bounds on these quantities which also depend only on Bell basis measurement data, thus making them readily accessible experimentally. 
 
This work is organized as follows. We first construct unbiased estimators, which depend only on Bell basis measurement data, for all entanglement measures computable using the parallelized c-SWAP test, thus recovering all results in Refs. \cite{foulds2021controlled,beckey2021computable} while using fewer resources. We then derive expressions showing how many measurement repetitions are needed to obtain an $\epsilon$-close approximation of these measures with high probability. Next, we extend the CEs to mixed states and introduce a family of lower bounds for the mixed state CEs which allow one to probe the multipartite entanglement of mixed quantum states, thus generalizing Refs. \cite{mintert2005concurrence,mintert2007observable,aolita2008scalable,beckey2021computable}. Finally, we demonstrate the feasibility of our methods by carrying out realistic, noisy experiments on a simulated Rydberg system. Background material, proofs, and simulation details can be found in the Supplementary Material~\cite{seeSupplementary}.

\textit{CEs via Parallelized c-SWAP circuit.}
To appreciate the utility of the Bell basis measurement scheme, one must first understand the CEs and how they can be estimated via the parallelized c-SWAP test. Thus, we begin by defining the CEs.

Let $\ket{\psi} \in (\mathbb{C}^2)^{\otimes n}$ denote a pure state of $n$-qubits. Further, denote the set of labels of the qubits as $\SC=\{1,2,\dots, n\}$. Throughout, we will let $s \subseteq \SC$ be any subset of the $n$ qubits with $\PC(s)$ the associated power set (i.e. the set of all subsets of $s$, which has cardinality $2^{|s|}$). With our notations in place, we can define the CE.

\begin{definition}[Ref. \cite{beckey2021computable}] \label{def:CE}
For any non-empty set of qubit labels $s\in \PC(\SC) \setminus \{\emptyset\}$, the Concentratable Entanglement is defined as 
\begin{align}\label{eq:CE}
    \CC_{\ket{\psi}}(s) &= 1 -\frac{1}{2^{|s|}} \sum_{\alpha \in \PC(s)} \tr{\rho_{\alpha}^2},
\end{align}
where the $\rho_{\alpha}$'s are reduced states of $\ket{\psi}\bra{\psi}$ obtained by tracing out subsystems with labels not in $\alpha$. For the trivial subset, we take $\tr{\rho_{\emptyset}^2}:=1$.
\end{definition}
When $s=\SC$, the sum in Def. \ref{def:CE} is simply a uniform average of subsystem purities. This matches the intuition that highly entangled pure states should have highly mixed (low purity) reduced states. Although many interesting properties of the CE are summarized in Ref. \cite{beckey2021computable}, we need only one more detail to motivate this current work. Namely, the fact that the CE can be estimated from the output probabilities of the parallelized c-SWAP test via
\begin{align}
    \CC_{\ket{\psi}}(s) &= 1- \sum_{\boldsymbol{z} \in \ZC_{\boldsymbol{0}}(s)} p(\boldsymbol{z}),
\end{align}
where $\boldsymbol{z} \in \{0,1\}^n$ denotes a length $n$ bitstring, $p(\boldsymbol{z})$ the probability of obtaining said bitstring, and $\ZC_{\boldsymbol{0}}(s)$ the set of all bitstrings with zeroes in the indices of $s$. As one can see from Fig. \ref{fig:SWAPn}, the parallelized c-SWAP test requires $3n$ qubits and $n$ Toffoli gates, which, on most platforms, must be further broken down into one- and two-qubit gates \cite{shende2008on}. Although some hardware platforms, like Rydberg atoms, can implement Toffoli gates natively with high fidelity \cite{pelegri2022high}, it would be preferable to eliminate the $3$-qubit gates altogether. This is exactly what the Bell basis method achieves while simultaneously reducing the qubit requirements from $3n$ to $2n$. Before seeing how this is done, we introduce some background on Bell basis measurements and introduce the required notation.

\begin{figure}[t!]
\centering
\includegraphics[width=0.9\columnwidth]{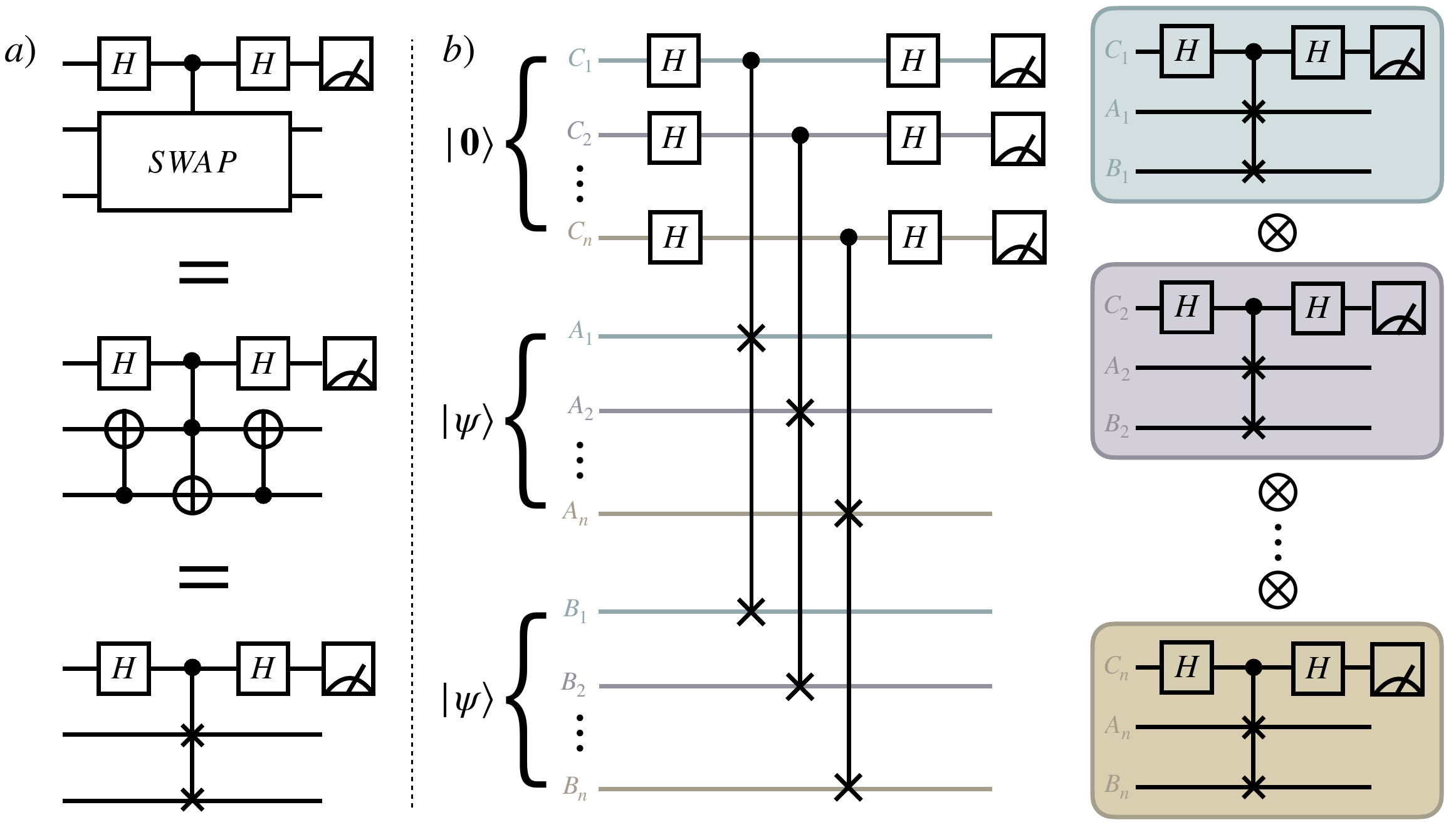}
\caption{\textbf{c-SWAP circuits.} a) Equivalent representations of the single qubit controlled-SWAP circuit. b) The $n$-qubit parallelized c-SWAP circuit can be used to probe a pure state $\ket{\psi}$'s entanglement \cite{foulds2021controlled,beckey2021computable}.}
\label{fig:SWAPn}
\end{figure} 
\textit{Bell basis measurements.}
Suppose we carry out $M$ rounds of Bell basis measurements. For each round $m \in \{1,\dots,M\}$, this consists of performing a Bell basis measurement on the $k$-th test and copy qubit for each $k \in \{1,\dots,n\}$, as shown pictorially in Fig. \ref{fig:bellBasis}. Measuring the $k$-th test and copy qubit in the Bell basis results in one of the four Bell states as the post measurement state 
\begin{align}
    B_k^{(m)} \in \left\{ \ket{\Phi^+}\bra{\Phi^+},\ket{\Phi^-}\bra{\Phi^-}, \ket{\Psi^+}\bra{\Psi^+},\ket{\Psi^-}\bra{\Psi^-}\right\}.
\end{align}
For our purposes, we consider $B_k^{(m)}$ as a random variable that takes values in the set of Bell basis projectors. For each of the $M$ rounds, we efficiently store the qubit label, $k$, and the corresponding measurement outcome $B_k^{(m)}$ in classical memory, which one can then post-process in a number of ways to obtain many entanglement measures of interest, as we will show.

To understand the power of Bell basis measurements, first note that the Bell states are eigenstates of the SWAP operator with eigenvalues $\pm 1$, so $\tr{\mathbb{F}_k B_k^{(m)}} = \pm 1$ for all $k,m$, where $\mathbb{F}_k$ is the SWAP operator acting on the $k$-th test and copy qubit. This connection between the SWAP operator and the Bell basis is why SWAP tests can be simulated by Bell basis measurement methods. For instance, one can construct an unbiased estimator of the purity of a single qubit as
\begin{align}\label{eq:purtiy-from-Bell}
       \E{ \frac{1}{M} \sum_{m=1}^M \tr{\mathbb{F} B^{(m)}}} &= \tr{\rho^2},
\end{align}
where the expectation is taken with respect to the empirical distribution resulting from Bell basis measurement outcomes. This method has been utilized by many experimental groups to estimate quantum purities \cite{daley2012measuring,islam2015measuring,bluvstein2022quantum,kaufman2016quantum}. In fact, from the data in Refs.~\cite{daley2012measuring,islam2015measuring,bluvstein2022quantum,kaufman2016quantum}, one could estimate \textit{all possible subsystem purities} of $n$-qubit states by extending the idea in Eq.~\ref{eq:purtiy-from-Bell}~\cite{garcia-escartin2013swap,seeSupplementary}. 

%Naively, one might try to estimate Eq.~\ref{eq:CE} by sampling subsystems at random, estimating these purities using the multipartite generalization of Eq.~\ref{eq:purtiy-from-Bell}, and using these estimates to compute the average subsystem purity on which the CE directly depends. Unfortunately, proving convergence of this method is intractable because 1) the samples would not be independent in general and 2) it would require the estimation of the average of many already averaged quantities. Both of these factors render standard concentration inequalities inapplicable. In contrast, standard statistical methods apply to our results. It is this simpler method, the main result of the present paper, to which we now turn.%

\textit{Multipartite entanglement from Bell basis measurement data.}
Our main results are concerned with the ancilla-free simulation of the parallelized c-SWAP test. The following theorems show how to recover all results of the c-SWAP test without the need for ancillary qubits or Toffoli gates, thus making the resulting entanglement measures far more experimentally accessible.

First, we show the existence of a family of unbiased estimators for the CEs which depend solely on Bell basis measurement outcomes.
\begin{theorem}
The quantities 
\begin{align}
    \hat{\mathcal{C}}_{\ket{\psi}}(s) &= 1- \frac{1}{M}\sum_{m=1}^M \prod_{k\in s} \left(\frac{1+\tr{\mathbb{F}_k B^{(m)}_k}}{2} \right),
\end{align}
are unbiased estimators of the Concentratable Entanglements. That is, for all $s \subseteq \SC$, 
\begin{align}
    \E{\hat{\mathcal{C}}_{\ket{\psi}}(s)} &= 1 - \frac{1}{2^{|s|}}\sum_{\alpha \in \mathcal{P}(s)} \text{tr}[\rho_{\alpha}^2],
\end{align}
where the expectation value is with respect to the probability distribution induced by the Bell basis measurements. 
\end{theorem}
This theorem implies that the circuit in Fig. \ref{fig:SWAPn}b can be completely simulated by a projective Bell basis measurement on two copies of a state of interest. Many well-known entanglement measures can be estimated using this result. By letting $s=\{j\}$, one obtains an estimate of $\frac{1}{2}(1-\tr{\rho_j^2})$, which, when averaged over all $j \in \SC$ yields the entanglement measure from Refs. \cite{meyer2002global,brennen2003observable}. At the opposite extreme, when $s= \SC$, one obtains a CE which is related to the generalized concurrence $c_n(\ket{\psi})$, as defined in Ref. \cite{carvalho2004decoherence,aolita2006measuring}, via the simple formula $\CC_{\ket{\psi}}(\SC) = c_n(\ket{\psi})^2/4$. This realization implies that the entanglement measure being explored in Ref. \cite{foulds2021controlled} was exactly the generalized concurrence as defined in Ref. \cite{carvalho2004decoherence}. In between these two extremes, many other well-defined measures of multipartite entanglement can be estimated, \textit{all from the same measurement data.}    

There is still more one can learn from Bell basis measurement data, however. For instance, we can state a very similar theorem for the $n$-tangle, another well-studied multipartite entanglement measure \cite{wong2001potential}.

\begin{theorem}
The quantity 
\begin{align}
    \hat{\tau}_{(n)} &= \frac{2^n}{M}\sum_{m=1}^M \prod_{k=1}^{n} \left(\frac{1-\tr{\mathbb{F}_k B^{(m)}_k}}{2} \right),
\end{align}
is an unbiased estimator of the $n$-tangle. That is,
\begin{align}
    \E{\hat{\tau}_{(n)}} &= \tau_{(n)},
\end{align}
where the expectation value is with respect to the probability distribution induced by the Bell basis measurements. 
\end{theorem}

\begin{figure}[t!]
\centering
\includegraphics[width=0.9\columnwidth]{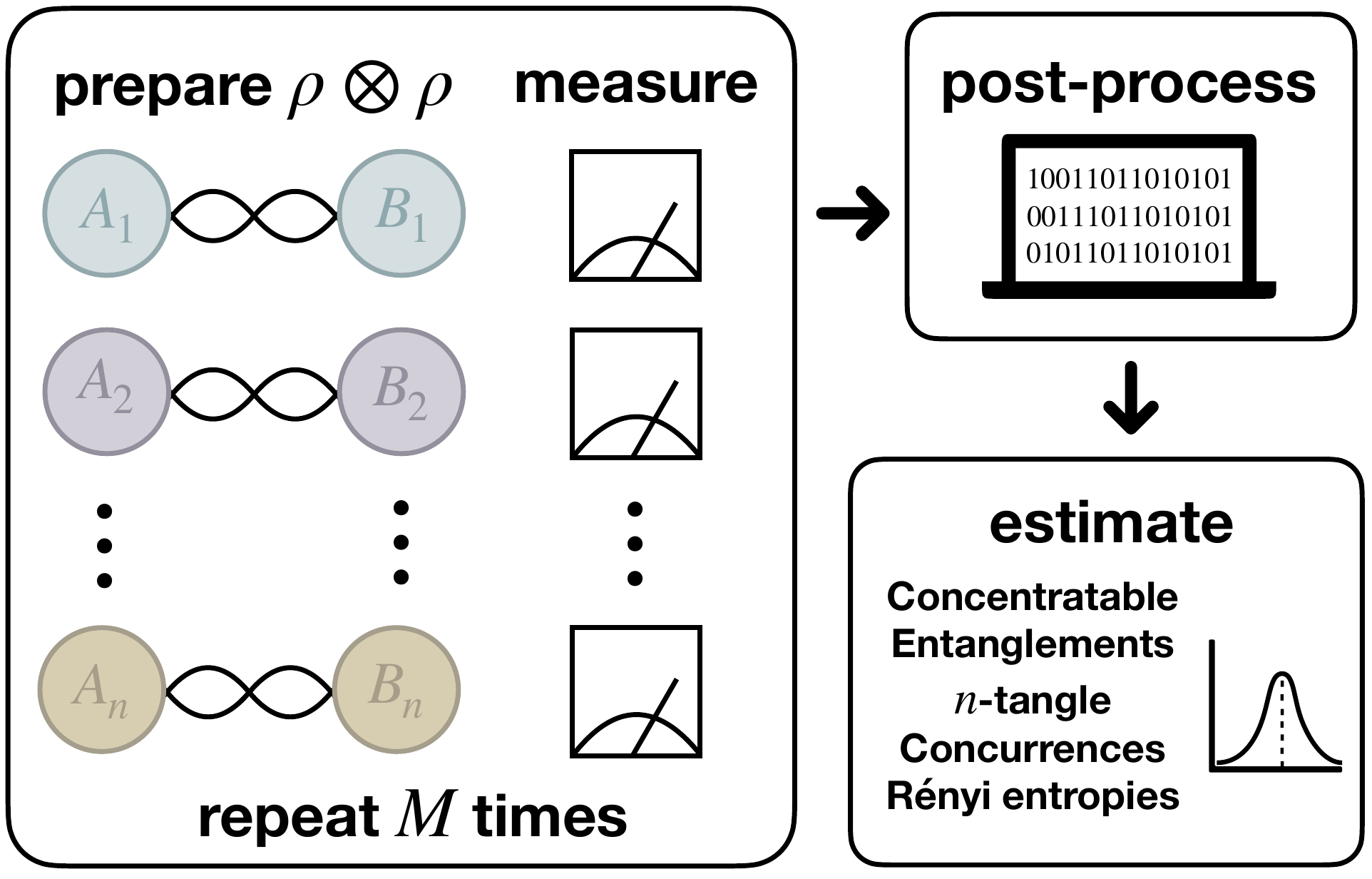}
\caption{\textbf{Bell Basis Estimation Method.} Experimentally, one first prepares the test state and an, ideally identical, copy state. We denote the composite state $\rho \otimes \rho$. Then the $k$-th subsystems in the test and copy states are entangled using native one- and two-qubit gates. This converts a computational basis measurement to a Bell basis measurement. The data from $M$ rounds of this procedure is stored in classical memory which one can then post-process in a number of ways to obtain many entanglement measures of interest.}
\label{fig:bellBasis}
\end{figure} 

With this theorem, we have recovered all of the measures shown in Ref. \cite{beckey2021computable} to be computable with the parallelized c-SWAP test. In addition to requiring fewer experimental resources, it is simple to determine how many rounds of Bell basis measurements are needed to achieve an $\epsilon$-close approximation of the estimators we have introduced. We formalize this statement in the following proposition, the proof of which follows directly from Hoeffding's inequality from classical statistics. 

\begin{proposition}\label{prop:confidence-intervals}
    Let $\epsilon, \delta > 0$ and $M=\Theta \left({\frac{\log{1/\delta}}{\epsilon^2}}\right)$. Further, let $\theta \in \{\mathcal{C}_{\ket{\psi}}(s), \tau_{(n)}\}$ and let $\hat{\theta}$ denote the corresponding estimator for $\theta$. Then we have 
    \begin{align}
        \left|  \hat{\theta} - \theta \right| < \epsilon,
    \end{align}
    with probability at least $1-\delta$.
\end{proposition}
This result, while simple and analytical, does not take into account the underlying probability distribution, and is thus not as tight as it could be. As we show in Fig. \ref{fig:BellMeasurements}(b), using information about the underlying distribution, one finds numerically that Prop.~\ref{prop:confidence-intervals} often leads to overestimates on the number of measurements needed to obtain $\epsilon$-close estimates of the quantities of interest.

Thus far, we have only considered estimating these measures given two identical copies of a \textit{pure} quantum state. In order for these methods to be truly useful on today's hardware, we must extend to the measures to mixed states.

\textit{CE for mixed states.}
The standard method of extending pure state entanglement measures to mixed states is a so-called \textit{convex-roof extension} \cite{bennett1996mixed,uhlmann2010roofs}

\begin{align}\label{eq:mixed-CE}
    \CC_{\rho}(s) &= \inf \sum_i p_i \CC_{\ket{\psi_i}}(s),
\end{align}
where the infimum is over the set of decompositions of the form $\rho=\sum_i p_i \ket{\psi_i}\bra{\psi_i}$, with $\sum_i p_i =1$. Because this optimization is generally difficult, we would like to avoid it. An alternative method is to find lower bounds for the mixed state CEs that depend only on Bell basis measurement outcomes. This allows one to bound the mixed state entanglement within the above framework developed for estimating pure state entanglement. 

We will construct the lower bounds on $\CC_{\rho}(s)$ using the relationship between CEs and the bipartite concurrences $c_{\alpha}(\ket{\psi})$ \cite{wootters2001entanglement}, as well as a known lower bound for the mixed state bipartite concurrence \cite{mintert2007observable}. Specifically, any CE can be expressed in terms of bipartite concurrences as
\begin{align}
    \CC_{\ket{\psi}} (s) = \frac{1}{2^{|s|+1}} \sum_{\alpha} c_{\alpha}^2(\ket{\psi}),
\end{align}
where $c_{\alpha}(\ket{\psi}) := \sqrt{2(1-\tr{\rho_{\alpha}^2})}$. Then, because we can use the known lower bound for each bipartite concurrence in the sum, we can construct a lower bound for any CE of interest. This is a generalization of the method used in Ref. \cite{aolita2008scalable} in which the authors derive a lower bound on the mixed state multipartite concurrence. For example, the lower bound on $\CC_{\rho}(\SC)$ takes the form
\begin{align}
    \CC_{\rho}^{\ell}(\SC) = \frac{1}{2^n} + (1-\frac{1}{2^n})\tr{\rho^2} - \frac{1}{2^n}\sum_{\alpha \in \PC (\SC)} \tr{\rho_{\alpha}^2}.
\end{align}
Because each term in this expression can be directly estimated from Bell basis measurement data, it allows one to quantify mixed state entanglement in the same framework developed above for pure state entanglement. We further note that, for high-purity states that are common in today's state of the art experiments, this bound is very close to the pure state theoretical value, as shown in Fig.~\ref{fig:BellMeasurements}. This can be seen by noting that $C_{\ket{\psi}}(\SC) - \CC_{\rho}^{\ell}(\SC) = (1-2^{-n})(1-\tr{\rho^2}),$ which is very close to zero for nearly pure states \cite{seeSupplementary}. With these bounds in place, we turn to demonstrating the viability of our proposed scheme via realistic Rydberg system simulations.

\textit{Rydberg atom simulations with noise.}
In Fig. \ref{fig:BellvsSWAP}(a) we illustrate the architectures that we propose for quantifying the CE using the c-SWAP test and the Bell basis measurement method in neutral atom systems. The c-SWAP circuit is implemented by arranging each group of atomic qubits $\{A_k,B_k,C_k\}$ in an equilateral triangle, in such a way that $CZ$ and $CCZ$ gates can be realized using the Rydberg pulse sequences described in \cite{pelegri2022high}. These global unitaries are then transformed to CNOT and Toffoli gates through the application of Hadamard gates to the target qubit before and after the Rydberg pulses. The Bell basis measurements are performed by applying Hadamard and CNOT gates to the relevant pairs of qubits $\{A_k,B_k\}$ and then measuring in the computational basis. We model the presence of experimental imperfections by substituting the ideal $CZ$ and $CCZ$ gates by non-unitary transformations~\cite{seeSupplementary}. The application of these imperfect gates on pure states results in phase errors and loss of norm, which mimics the leakage of population outside of the computational basis under the application of the Rydberg pulses. Since occurences of leakage can be detected and discarded in the post-processing of the experimental data, we re-normalize the state resulting from the application of the non-unitary gates before computing its CE. We keep track of the loss of norm for the purpose of estimating the number of repetitions required to achieve a desired accuracy. For simplicity of notation, we denote the CE computed over renormalized pure states as $\mathcal{C}(\mathcal{S})$.
\begin{figure}[t!]
\centering
\includegraphics[width=0.9\columnwidth]{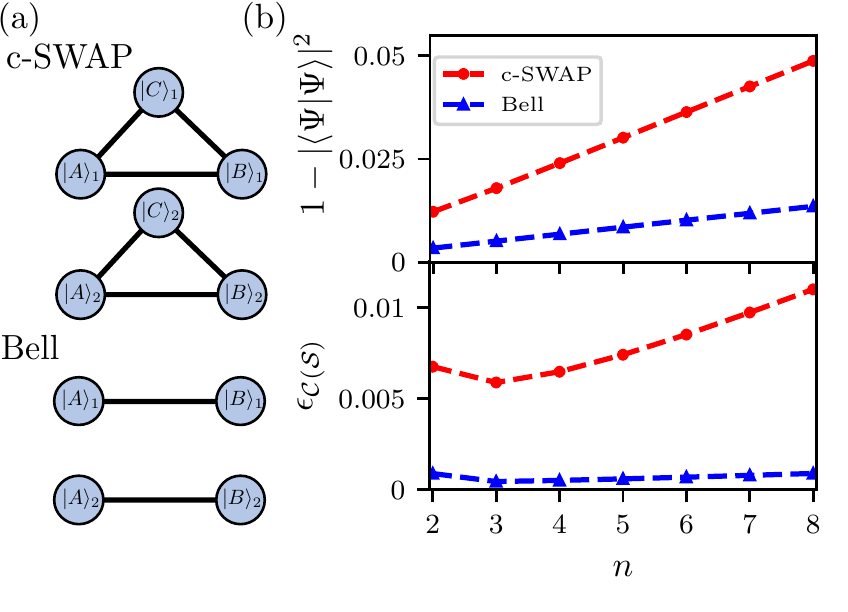}
\caption{\textbf{c-SWAP vs. Bell Basis Method.} a) Pictorial representation of the geometry Rydberg atoms would be placed in to implement either the c-SWAP or Bell methods for a two-qubit state example. b) Top panel indicates the loss of norm due to imperfect Rydberg pulses. Bottom panel shows the error that this causes in estimation of the CE of a GHZ state. In both cases, we see the Bell basis method outperforms the c-SWAP.}
\label{fig:BellvsSWAP}
\end{figure} 

The Bell measurement method offers a substantial practical advantage with respect to the c-SWAP test for estimating $\mathcal{C}(\mathcal{S})$ due to its reduced requirement on the number of copies and its significantly lower total gate count. To illustrate this, in Fig.~\ref{fig:BellvsSWAP} we compare the results obtained when measuring with both methods the CE for an $n$-qubit Greenberger–Horne–Zeilinger (GHZ) state ${\ket{\mathrm{GHZ}_n}=\frac{1}{\sqrt{2}}\left(\ket{0}^{\otimes n}+\ket{1}^{\otimes n}\right)}$. In the lower plot of Fig.~\ref{fig:BellvsSWAP}(b) we show the relative discrepancy $\epsilon_{\CC(\SC)}$ between the value of $\CC(\SC)$ obtained with each method and the analytical result \cite{foulds2021controlled,beckey2021computable} as a function of $n$. We observe that for all numbers of qubits the Bell measurement method yields more accurate results than the c-SWAP test due to the reduction in accumulated phase errors. The upper plot of Fig.~\ref{fig:BellvsSWAP}(b) shows that the loss of norm is smaller for the Bell measurement method than for the c-SWAP test, meaning that the former method would require fewer repetitions to achieve a given level of accuracy than the latter.

\begin{figure}[t!]
\centering
\includegraphics[width=0.9\columnwidth]{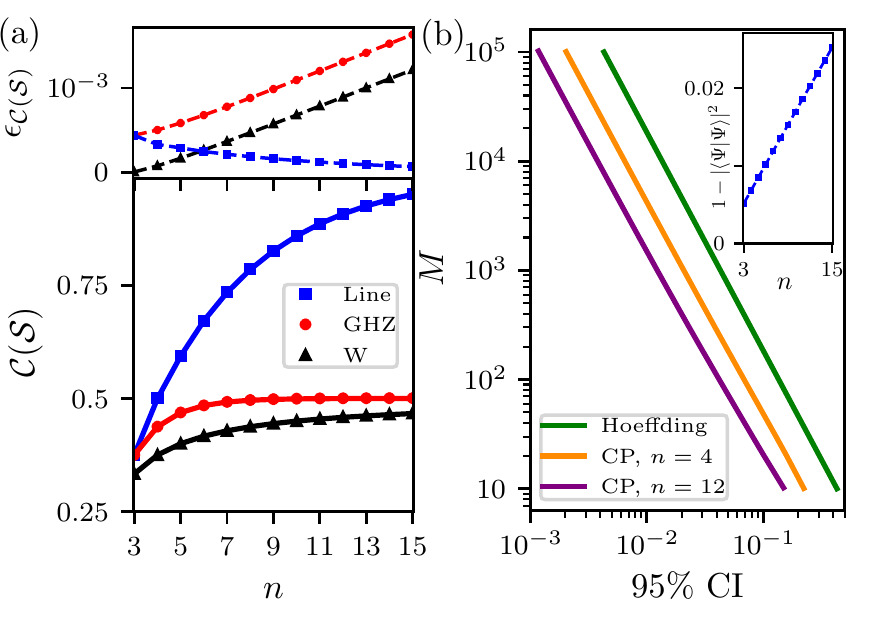}
\caption{\textbf{Bell basis method with realistic Rydberg gates} a) Bottom panel: CE of GHZ, W and Line states, with solid lines indicating theoretical values and dots representing the results obtained with noisy Rydberg gates. Top panel: relative discrepancy between the theoretical and simulated values of the CE. b) Number of measurements required as a function of the desired size of the $95\%$ CI on the precision of the estimation of the CE for a Line state with different numbers of qubits. The inset shows the loss of norm for this state when the CE is estimated with the Bell basis method.}
\label{fig:BellMeasurements}
\end{figure} 

Having established the superiority of the Bell measurement method in the presence of experimental imperfections, we turn to investigating its performance for estimating $\mathcal{C}(\mathcal{S})$ for different classes of highly entangled states. In the lower plot of Fig.~\ref{fig:BellMeasurements}(a) we show the theoretical (solid lines) and simulated experimental (dots) values of $\CC(\SC)$ as a function of the number of qubits $n$ for GHZ, W and Line states (all of which admit analytical formulas which are given in the Supplementary Materials). We observe that values of $\CC(\SC)$ remain clearly distinguishable between the three states up to $n=15$. Furthermore, as illustrated in the upper plot the relative discrepancy between the theoretical and simulated values remains $\epsilon_{\CC(\SC)}\lesssim 10^{-3}$ for the range of $n$ and states considered. In Fig.~\ref{fig:BellMeasurements}b we show the size of the $95\%$ Confidence Interval (CI) in the Maximum Likelihood Estimation of the $\mathcal{C}(\mathcal{S})$ for a Line state of $n=4,12$ qubits computed with the Clopper-Pearson (CP) method as a function of the total number of measurements $M$, as well as the bound provided by Hoeffding's inequality. The CP method predicts a lower requirement in the number of measurements to achieve a given size of the CI because it is tailored to the binomial probability distribution that governs the statistics of $\mathcal{C}(\mathcal{S})$ measurements, but Hoeffding's inequality provides a useful bound which is easy to compute analytically. The inset of Fig.~\ref{fig:BellMeasurements}(b) shows the loss of norm as a function of the number of qubits. Even for $n=15$ the norm of the state remains $|\braket{\Psi}{\Psi}|^2\sim 0.98$, meaning that the number of experiment repetitions would only need to be increased by $\lesssim 2\%$ to make up for the leakage outside of the computational basis.

\textit{Conclusion.} 
We have shown how to estimate the CEs and $n$-tangle from Bell basis measurement data. We extended the definition of the CEs to mixed states and showed how to estimate lower bounds on the mixed state CE from Bell basis measurement data. Our methods simultaneously make these measures more experimentally accessible, while also simplifying their associated theoretical analysis.

An interesting direction for future work would be to compare, in terms of both theoretical sample complexity and performance on real hardware, this Bell basis method to local randomized measurements~\cite{vanenk2012measuring,elben2019statistical,brydges2019probing,ohnemus2022quantifying,notarnicola2021randomized,rath2021importance} and classical shadows~\cite{huang2020predicting,elben2020mixed}, two modern techniques that have been applied to the study of other entanglement measures.

\section*{Acknowledgments}
JLB was initially supported by the National Science Foundation Graduate Research Fellowship under Grant No.~1650115 and was partially supported by NSF grant 1915407. This material is based upon work supported by the U.S. Department of Energy, Office of Science, National Quantum Information Science Research Centers, Quantum Systems Accelerator. JLB also acknowledges helpful discussions with Hsin-Yuan (Robert) Huang, Michael Walter, Graeme Smith, Guangkuo Liu, and Louis Schatzki. GP and NJP are supported by the EPSRC (Grant No.~EP/T005386/1) and M Squared Lasers Ltd. SF is supported by a UK EPSRC funded DTG studentship (ref 2210204) and thanks Tim Spiller and Viv Kendon for many useful discussions. GP and NJP acknowledge fruitful discussions with Jonathan Pritchard and Andrew Daley.  %The data that support the findings of this study is openly available at \cite{DOIpure}.%

\let\oldaddcontentsline\addcontentsline% Store \addcontentsline
\renewcommand{\addcontentsline}[3]{}% Make \addcontentsline a no-op
\bibliography{references.bib} 

\begin{thebibliography}{61}%
\makeatletter
\providecommand \@ifxundefined [1]{%
 \@ifx{#1\undefined}
}%
\providecommand \@ifnum [1]{%
 \ifnum #1\expandafter \@firstoftwo
 \else \expandafter \@secondoftwo
 \fi
}%
\providecommand \@ifx [1]{%
 \ifx #1\expandafter \@firstoftwo
 \else \expandafter \@secondoftwo
 \fi
}%
\providecommand \natexlab [1]{#1}%
\providecommand \enquote  [1]{``#1''}%
\providecommand \bibnamefont  [1]{#1}%
\providecommand \bibfnamefont [1]{#1}%
\providecommand \citenamefont [1]{#1}%
\providecommand \href@noop [0]{\@secondoftwo}%
\providecommand \href [0]{\begingroup \@sanitize@url \@href}%
\providecommand \@href[1]{\@@startlink{#1}\@@href}%
\providecommand \@@href[1]{\endgroup#1\@@endlink}%
\providecommand \@sanitize@url [0]{\catcode `\\12\catcode `\$12\catcode
  `\&12\catcode `\#12\catcode `\^12\catcode `\_12\catcode `\%12\relax}%
\providecommand \@@startlink[1]{}%
\providecommand \@@endlink[0]{}%
\providecommand \url  [0]{\begingroup\@sanitize@url \@url }%
\providecommand \@url [1]{\endgroup\@href {#1}{\urlprefix }}%
\providecommand \urlprefix  [0]{URL }%
\providecommand \Eprint [0]{\href }%
\providecommand \doibase [0]{http://dx.doi.org/}%
\providecommand \selectlanguage [0]{\@gobble}%
\providecommand \bibinfo  [0]{\@secondoftwo}%
\providecommand \bibfield  [0]{\@secondoftwo}%
\providecommand \translation [1]{[#1]}%
\providecommand \BibitemOpen [0]{}%
\providecommand \bibitemStop [0]{}%
\providecommand \bibitemNoStop [0]{.\EOS\space}%
\providecommand \EOS [0]{\spacefactor3000\relax}%
\providecommand \BibitemShut  [1]{\csname bibitem#1\endcsname}%
\let\auto@bib@innerbib\@empty
%</preamble>
\bibitem [{\citenamefont {Horodecki}\ \emph {et~al.}(2009)\citenamefont
  {Horodecki}, \citenamefont {Horodecki}, \citenamefont {Horodecki},\ and\
  \citenamefont {Horodecki}}]{horodecki2009quantum}%
  \BibitemOpen
  \bibfield  {author} {\bibinfo {author} {\bibfnamefont {Ryszard}\ \bibnamefont
  {Horodecki}}, \bibinfo {author} {\bibfnamefont {Pawe\l{}}\ \bibnamefont
  {Horodecki}}, \bibinfo {author} {\bibfnamefont {Micha\l{}}\ \bibnamefont
  {Horodecki}}, \ and\ \bibinfo {author} {\bibfnamefont {Karol}\ \bibnamefont
  {Horodecki}},\ }\bibfield  {title} {\enquote {\bibinfo {title} {Quantum
  entanglement},}\ }\href {\doibase 10.1103/RevModPhys.81.865} {\bibfield
  {journal} {\bibinfo  {journal} {Rev. Mod. Phys.}\ }\textbf {\bibinfo {volume}
  {81}},\ \bibinfo {pages} {865--942} (\bibinfo {year} {2009})}\BibitemShut
  {NoStop}%
\bibitem [{\citenamefont {Friis}\ \emph {et~al.}(2018)\citenamefont {Friis},
  \citenamefont {Vitagliano}, \citenamefont {Malik},\ and\ \citenamefont
  {Huber}}]{friis2018entanglement}%
  \BibitemOpen
  \bibfield  {author} {\bibinfo {author} {\bibfnamefont {Nicolai}\ \bibnamefont
  {Friis}}, \bibinfo {author} {\bibfnamefont {Giuseppe}\ \bibnamefont
  {Vitagliano}}, \bibinfo {author} {\bibfnamefont {Mehul}\ \bibnamefont
  {Malik}}, \ and\ \bibinfo {author} {\bibfnamefont {Marcus}\ \bibnamefont
  {Huber}},\ }\bibfield  {title} {\enquote {\bibinfo {title} {Entanglement
  certification from theory to experiment},}\ }\href {\doibase
  10.1038/s42254-018-0003-5} {\bibfield  {journal} {\bibinfo  {journal} {Nature
  Reviews Physics}\ }\textbf {\bibinfo {volume} {1}},\ \bibinfo {pages}
  {72–87} (\bibinfo {year} {2018})}\BibitemShut {NoStop}%
\bibitem [{\citenamefont {Vedral}\ and\ \citenamefont
  {Kashefi}(2002)}]{vedral2002uniqueness}%
  \BibitemOpen
  \bibfield  {author} {\bibinfo {author} {\bibfnamefont {Vlatko}\ \bibnamefont
  {Vedral}}\ and\ \bibinfo {author} {\bibfnamefont {Elham}\ \bibnamefont
  {Kashefi}},\ }\bibfield  {title} {\enquote {\bibinfo {title} {Uniqueness of
  the entanglement measure for bipartite pure states and thermodynamics},}\
  }\href {\doibase 10.1103/PhysRevLett.89.037903} {\bibfield  {journal}
  {\bibinfo  {journal} {Phys. Rev. Lett.}\ }\textbf {\bibinfo {volume} {89}},\
  \bibinfo {pages} {037903} (\bibinfo {year} {2002})}\BibitemShut {NoStop}%
\bibitem [{\citenamefont {Haah}\ \emph {et~al.}(2017)\citenamefont {Haah},
  \citenamefont {Harrow}, \citenamefont {Ji}, \citenamefont {Wu},\ and\
  \citenamefont {Yu}}]{haah2017sample}%
  \BibitemOpen
  \bibfield  {author} {\bibinfo {author} {\bibfnamefont {Jeongwan}\
  \bibnamefont {Haah}}, \bibinfo {author} {\bibfnamefont {Aram~W.}\
  \bibnamefont {Harrow}}, \bibinfo {author} {\bibfnamefont {Zhengfeng}\
  \bibnamefont {Ji}}, \bibinfo {author} {\bibfnamefont {Xiaodi}\ \bibnamefont
  {Wu}}, \ and\ \bibinfo {author} {\bibfnamefont {Nengkun}\ \bibnamefont
  {Yu}},\ }\bibfield  {title} {\enquote {\bibinfo {title} {Sample-optimal
  tomography of quantum states},}\ }\href {\doibase 10.1109/TIT.2017.2719044}
  {\bibfield  {journal} {\bibinfo  {journal} {IEEE Transactions on Information
  Theory}\ }\textbf {\bibinfo {volume} {63}},\ \bibinfo {pages} {5628--5641}
  (\bibinfo {year} {2017})}\BibitemShut {NoStop}%
\bibitem [{\citenamefont {O'Donnell}\ and\ \citenamefont
  {Wright}(2015)}]{odonnell2015efficient}%
  \BibitemOpen
  \bibfield  {author} {\bibinfo {author} {\bibfnamefont {Ryan}\ \bibnamefont
  {O'Donnell}}\ and\ \bibinfo {author} {\bibfnamefont {John}\ \bibnamefont
  {Wright}},\ }\href {\doibase 10.48550/ARXIV.1508.01907} {\enquote {\bibinfo
  {title} {\href{https://arxiv.org/abs/1508.01907}{Efficient quantum
  tomography}},}\ } (\bibinfo {year} {2015})\BibitemShut {NoStop}%
\bibitem [{\citenamefont {Foulds}\ \emph {et~al.}(2021)\citenamefont {Foulds},
  \citenamefont {Kendon},\ and\ \citenamefont
  {Spiller}}]{foulds2021controlled}%
  \BibitemOpen
  \bibfield  {author} {\bibinfo {author} {\bibfnamefont {Steph}\ \bibnamefont
  {Foulds}}, \bibinfo {author} {\bibfnamefont {Viv}\ \bibnamefont {Kendon}}, \
  and\ \bibinfo {author} {\bibfnamefont {Tim}\ \bibnamefont {Spiller}},\
  }\bibfield  {title} {\enquote {\bibinfo {title} {The controlled swap test for
  determining quantum entanglement},}\ }\href {\doibase
  10.1088/2058-9565/abe458} {\bibfield  {journal} {\bibinfo  {journal} {Quantum
  Science and Technology}\ }\textbf {\bibinfo {volume} {6}},\ \bibinfo {pages}
  {035002} (\bibinfo {year} {2021})}\BibitemShut {NoStop}%
\bibitem [{\citenamefont {Beckey}\ \emph {et~al.}(2021)\citenamefont {Beckey},
  \citenamefont {Gigena}, \citenamefont {Coles},\ and\ \citenamefont
  {Cerezo}}]{beckey2021computable}%
  \BibitemOpen
  \bibfield  {author} {\bibinfo {author} {\bibfnamefont {Jacob~L.}\
  \bibnamefont {Beckey}}, \bibinfo {author} {\bibfnamefont {N.}~\bibnamefont
  {Gigena}}, \bibinfo {author} {\bibfnamefont {Patrick~J.}\ \bibnamefont
  {Coles}}, \ and\ \bibinfo {author} {\bibfnamefont {M.}~\bibnamefont
  {Cerezo}},\ }\bibfield  {title} {\enquote {\bibinfo {title} {Computable and
  operationally meaningful multipartite entanglement measures},}\ }\href
  {\doibase 10.1103/PhysRevLett.127.140501} {\bibfield  {journal} {\bibinfo
  {journal} {Phys. Rev. Lett.}\ }\textbf {\bibinfo {volume} {127}},\ \bibinfo
  {pages} {140501} (\bibinfo {year} {2021})}\BibitemShut {NoStop}%
\bibitem [{\citenamefont {Cullen}\ and\ \citenamefont
  {Kok}(2022)}]{cullen2022calculating}%
  \BibitemOpen
  \bibfield  {author} {\bibinfo {author} {\bibfnamefont {Alice~R.}\
  \bibnamefont {Cullen}}\ and\ \bibinfo {author} {\bibfnamefont {Pieter}\
  \bibnamefont {Kok}},\ }\bibfield  {title} {\enquote {\bibinfo {title}
  {Calculating concentratable entanglement in graph states},}\ }\href
  {https://arxiv.org/abs/2207.11997} {\bibfield  {journal} {\bibinfo  {journal}
  {arXiv preprint arXiv:2207.11997}\ } (\bibinfo {year} {2022})}\BibitemShut
  {NoStop}%
\bibitem [{\citenamefont {Schatzki}\ \emph {et~al.}(2021)\citenamefont
  {Schatzki}, \citenamefont {Arrasmith}, \citenamefont {Coles},\ and\
  \citenamefont {Cerezo}}]{schatzki2021entangled}%
  \BibitemOpen
  \bibfield  {author} {\bibinfo {author} {\bibfnamefont {Louis}\ \bibnamefont
  {Schatzki}}, \bibinfo {author} {\bibfnamefont {Andrew}\ \bibnamefont
  {Arrasmith}}, \bibinfo {author} {\bibfnamefont {Patrick~J.}\ \bibnamefont
  {Coles}}, \ and\ \bibinfo {author} {\bibfnamefont {M.}~\bibnamefont
  {Cerezo}},\ }\bibfield  {title} {\enquote {\bibinfo {title} {Entangled
  datasets for quantum machine learning},}\ }\href
  {https://arxiv.org/abs/2109.03400} {\bibfield  {journal} {\bibinfo  {journal}
  {arXiv preprint arXiv:2109.03400}\ } (\bibinfo {year} {2021})}\BibitemShut
  {NoStop}%
\bibitem [{\citenamefont {Schatzki}\ \emph {et~al.}(2022)\citenamefont
  {Schatzki}, \citenamefont {Liu}, \citenamefont {Cerezo},\ and\ \citenamefont
  {Chitambar}}]{schatzki2022hierarchy}%
  \BibitemOpen
  \bibfield  {author} {\bibinfo {author} {\bibfnamefont {Louis}\ \bibnamefont
  {Schatzki}}, \bibinfo {author} {\bibfnamefont {Guangkuo}\ \bibnamefont
  {Liu}}, \bibinfo {author} {\bibfnamefont {M.}~\bibnamefont {Cerezo}}, \ and\
  \bibinfo {author} {\bibfnamefont {Eric}\ \bibnamefont {Chitambar}},\
  }\bibfield  {title} {\enquote {\bibinfo {title} {A hierarchy of multipartite
  correlations based on concentratable entanglement},}\ }\href
  {https://arxiv.org/abs/2209.07607} {\bibfield  {journal} {\bibinfo  {journal}
  {arXiv preprint arXiv:2209.07607}\ } (\bibinfo {year} {2022})}\BibitemShut
  {NoStop}%
\bibitem [{\citenamefont {Wong}\ and\ \citenamefont
  {Christensen}(2001)}]{wong2001potential}%
  \BibitemOpen
  \bibfield  {author} {\bibinfo {author} {\bibfnamefont {Alexander}\
  \bibnamefont {Wong}}\ and\ \bibinfo {author} {\bibfnamefont {Nelson}\
  \bibnamefont {Christensen}},\ }\bibfield  {title} {\enquote {\bibinfo {title}
  {Potential multiparticle entanglement measure},}\ }\href {\doibase
  10.1103/PhysRevA.63.044301} {\bibfield  {journal} {\bibinfo  {journal} {Phys.
  Rev. A}\ }\textbf {\bibinfo {volume} {63}},\ \bibinfo {pages} {044301}
  (\bibinfo {year} {2001})}\BibitemShut {NoStop}%
\bibitem [{\citenamefont {Prove}\ \emph {et~al.}(2022)\citenamefont {Prove},
  \citenamefont {Foulds},\ and\ \citenamefont {Kendon}}]{prove22}%
  \BibitemOpen
  \bibfield  {author} {\bibinfo {author} {\bibfnamefont {Oliver}\ \bibnamefont
  {Prove}}, \bibinfo {author} {\bibfnamefont {Steph}\ \bibnamefont {Foulds}}, \
  and\ \bibinfo {author} {\bibfnamefont {Viv}\ \bibnamefont {Kendon}},\
  }\bibfield  {title} {\enquote {\bibinfo {title} {Generalizing the controlled
  swap test for entanglement for practical applications: qudit, optical, and
  slightly mixed states},}\ }\href {https://arxiv.org/abs/2112.04333}
  {\bibfield  {journal} {\bibinfo  {journal} {arXiv preprint arXiv:2112.04333}\
  } (\bibinfo {year} {2022})}\BibitemShut {NoStop}%
\bibitem [{\citenamefont {Saffman}\ \emph {et~al.}(2010)\citenamefont
  {Saffman}, \citenamefont {Walker},\ and\ \citenamefont
  {M\o{}lmer}}]{saffman2010quantum}%
  \BibitemOpen
  \bibfield  {author} {\bibinfo {author} {\bibfnamefont {M.}~\bibnamefont
  {Saffman}}, \bibinfo {author} {\bibfnamefont {T.~G.}\ \bibnamefont {Walker}},
  \ and\ \bibinfo {author} {\bibfnamefont {K.}~\bibnamefont {M\o{}lmer}},\
  }\bibfield  {title} {\enquote {\bibinfo {title} {Quantum information with
  rydberg atoms},}\ }\href {\doibase 10.1103/RevModPhys.82.2313} {\bibfield
  {journal} {\bibinfo  {journal} {Rev. Mod. Phys.}\ }\textbf {\bibinfo {volume}
  {82}},\ \bibinfo {pages} {2313--2363} (\bibinfo {year} {2010})}\BibitemShut
  {NoStop}%
\bibitem [{\citenamefont {Adams}\ \emph {et~al.}(2019)\citenamefont {Adams},
  \citenamefont {Pritchard},\ and\ \citenamefont {Shaffer}}]{adams2019rydberg}%
  \BibitemOpen
  \bibfield  {author} {\bibinfo {author} {\bibfnamefont {C~S}\ \bibnamefont
  {Adams}}, \bibinfo {author} {\bibfnamefont {J~D}\ \bibnamefont {Pritchard}},
  \ and\ \bibinfo {author} {\bibfnamefont {J~P}\ \bibnamefont {Shaffer}},\
  }\bibfield  {title} {\enquote {\bibinfo {title} {Rydberg atom quantum
  technologies},}\ }\href {\doibase 10.1088/1361-6455/ab52ef} {\bibfield
  {journal} {\bibinfo  {journal} {Journal of Physics B: Atomic, Molecular and
  Optical Physics}\ }\textbf {\bibinfo {volume} {53}},\ \bibinfo {pages}
  {012002} (\bibinfo {year} {2019})}\BibitemShut {NoStop}%
\bibitem [{\citenamefont {Brion}\ \emph {et~al.}(2007)\citenamefont {Brion},
  \citenamefont {Mouritzen},\ and\ \citenamefont {M\o{}lmer}}]{brion07a}%
  \BibitemOpen
  \bibfield  {author} {\bibinfo {author} {\bibfnamefont {E.}~\bibnamefont
  {Brion}}, \bibinfo {author} {\bibfnamefont {A.~S.}\ \bibnamefont
  {Mouritzen}}, \ and\ \bibinfo {author} {\bibfnamefont {K.}~\bibnamefont
  {M\o{}lmer}},\ }\bibfield  {title} {\enquote {\bibinfo {title} {{Conditional
  dynamics induced by new configurations for Rydberg dipole-dipole
  interactions}},}\ }\href {\doibase 10.1103/PhysRevA.76.022334} {\bibfield
  {journal} {\bibinfo  {journal} {Phys. Rev. A}\ }\textbf {\bibinfo {volume}
  {76}},\ \bibinfo {pages} {022334} (\bibinfo {year} {2007})}\BibitemShut
  {NoStop}%
\bibitem [{\citenamefont {Isenhower}\ \emph {et~al.}(2011)\citenamefont
  {Isenhower}, \citenamefont {Saffman},\ and\ \citenamefont
  {M{\o}lmer}}]{isenhower11}%
  \BibitemOpen
  \bibfield  {author} {\bibinfo {author} {\bibfnamefont {L.}~\bibnamefont
  {Isenhower}}, \bibinfo {author} {\bibfnamefont {M.}~\bibnamefont {Saffman}},
  \ and\ \bibinfo {author} {\bibfnamefont {K.}~\bibnamefont {M{\o}lmer}},\
  }\bibfield  {title} {\enquote {\bibinfo {title} {{Multibit $C_k$-NOT quantum
  gates via Rydberg blockade}},}\ }\href {\doibase 10.1007/s11128-011-0292-4}
  {\bibfield  {journal} {\bibinfo  {journal} {Quantum Information Processing}\
  }\textbf {\bibinfo {volume} {10}},\ \bibinfo {pages} {755} (\bibinfo {year}
  {2011})}\BibitemShut {NoStop}%
\bibitem [{\citenamefont {Shi}(2018)}]{shi18}%
  \BibitemOpen
  \bibfield  {author} {\bibinfo {author} {\bibfnamefont {Xiao-Feng}\
  \bibnamefont {Shi}},\ }\bibfield  {title} {\enquote {\bibinfo {title}
  {{Deutsch, Toffoli, and cnot Gates via Rydberg Blockade of Neutral Atoms}},}\
  }\href {\doibase 10.1103/PhysRevApplied.9.051001} {\bibfield  {journal}
  {\bibinfo  {journal} {Phys. Rev. Applied}\ }\textbf {\bibinfo {volume} {9}},\
  \bibinfo {pages} {051001} (\bibinfo {year} {2018})}\BibitemShut {NoStop}%
\bibitem [{\citenamefont {Su}\ \emph {et~al.}(2017)\citenamefont {Su},
  \citenamefont {Gao}, \citenamefont {Liang},\ and\ \citenamefont
  {Zhang}}]{su18}%
  \BibitemOpen
  \bibfield  {author} {\bibinfo {author} {\bibfnamefont {Shi-Lei}\ \bibnamefont
  {Su}}, \bibinfo {author} {\bibfnamefont {Ya}~\bibnamefont {Gao}}, \bibinfo
  {author} {\bibfnamefont {Erjun}\ \bibnamefont {Liang}}, \ and\ \bibinfo
  {author} {\bibfnamefont {Shou}\ \bibnamefont {Zhang}},\ }\bibfield  {title}
  {\enquote {\bibinfo {title} {{Fast Rydberg antiblockade regime and its
  applications in quantum logic gates}},}\ }\href {\doibase
  10.1103/PhysRevA.95.022319} {\bibfield  {journal} {\bibinfo  {journal} {Phys.
  Rev. A}\ }\textbf {\bibinfo {volume} {95}},\ \bibinfo {pages} {022319}
  (\bibinfo {year} {2017})}\BibitemShut {NoStop}%
\bibitem [{\citenamefont {Su}\ \emph {et~al.}(2018)\citenamefont {Su},
  \citenamefont {Shen}, \citenamefont {Liang},\ and\ \citenamefont
  {Zhang}}]{su18a}%
  \BibitemOpen
  \bibfield  {author} {\bibinfo {author} {\bibfnamefont {S.~L.}\ \bibnamefont
  {Su}}, \bibinfo {author} {\bibfnamefont {H.~Z.}\ \bibnamefont {Shen}},
  \bibinfo {author} {\bibfnamefont {Erjun}\ \bibnamefont {Liang}}, \ and\
  \bibinfo {author} {\bibfnamefont {Shou}\ \bibnamefont {Zhang}},\ }\bibfield
  {title} {\enquote {\bibinfo {title} {{One-step construction of the
  multiple-qubit Rydberg controlled-phase gate}},}\ }\href {\doibase
  10.1103/PhysRevA.98.032306} {\bibfield  {journal} {\bibinfo  {journal} {Phys.
  Rev. A}\ }\textbf {\bibinfo {volume} {98}},\ \bibinfo {pages} {032306}
  (\bibinfo {year} {2018})}\BibitemShut {NoStop}%
\bibitem [{\citenamefont {Beterov}\ \emph {et~al.}(2018)\citenamefont
  {Beterov}, \citenamefont {Ashkarin}, \citenamefont {Yakshina}, \citenamefont
  {Tretyakov}, \citenamefont {Entin}, \citenamefont {Ryabtsev}, \citenamefont
  {Cheinet}, \citenamefont {Pillet},\ and\ \citenamefont
  {Saffman}}]{beterov18}%
  \BibitemOpen
  \bibfield  {author} {\bibinfo {author} {\bibfnamefont {I.~I.}\ \bibnamefont
  {Beterov}}, \bibinfo {author} {\bibfnamefont {I.~N.}\ \bibnamefont
  {Ashkarin}}, \bibinfo {author} {\bibfnamefont {E.~A.}\ \bibnamefont
  {Yakshina}}, \bibinfo {author} {\bibfnamefont {D.~B.}\ \bibnamefont
  {Tretyakov}}, \bibinfo {author} {\bibfnamefont {V.~M.}\ \bibnamefont
  {Entin}}, \bibinfo {author} {\bibfnamefont {I.~I.}\ \bibnamefont {Ryabtsev}},
  \bibinfo {author} {\bibfnamefont {P.}~\bibnamefont {Cheinet}}, \bibinfo
  {author} {\bibfnamefont {P.}~\bibnamefont {Pillet}}, \ and\ \bibinfo {author}
  {\bibfnamefont {M.}~\bibnamefont {Saffman}},\ }\bibfield  {title} {\enquote
  {\bibinfo {title} {{Fast three-qubit Toffoli quantum gate based on three-body
  F\"orster resonances in Rydberg atoms}},}\ }\href {\doibase
  10.1103/PhysRevA.98.042704} {\bibfield  {journal} {\bibinfo  {journal} {Phys.
  Rev. A}\ }\textbf {\bibinfo {volume} {98}},\ \bibinfo {pages} {042704}
  (\bibinfo {year} {2018})}\BibitemShut {NoStop}%
\bibitem [{\citenamefont {Levine}\ \emph {et~al.}(2019)\citenamefont {Levine},
  \citenamefont {Keesling}, \citenamefont {Semeghini}, \citenamefont {Omran},
  \citenamefont {Wang}, \citenamefont {Ebadi}, \citenamefont {Bernien},
  \citenamefont {Greiner}, \citenamefont {Vuleti\ifmmode~\acute{c}\else
  \'{c}\fi{}}, \citenamefont {Pichler},\ and\ \citenamefont
  {Lukin}}]{Levine2019parallel}%
  \BibitemOpen
  \bibfield  {author} {\bibinfo {author} {\bibfnamefont {Harry}\ \bibnamefont
  {Levine}}, \bibinfo {author} {\bibfnamefont {Alexander}\ \bibnamefont
  {Keesling}}, \bibinfo {author} {\bibfnamefont {Giulia}\ \bibnamefont
  {Semeghini}}, \bibinfo {author} {\bibfnamefont {Ahmed}\ \bibnamefont
  {Omran}}, \bibinfo {author} {\bibfnamefont {Tout~T.}\ \bibnamefont {Wang}},
  \bibinfo {author} {\bibfnamefont {Sepehr}\ \bibnamefont {Ebadi}}, \bibinfo
  {author} {\bibfnamefont {Hannes}\ \bibnamefont {Bernien}}, \bibinfo {author}
  {\bibfnamefont {Markus}\ \bibnamefont {Greiner}}, \bibinfo {author}
  {\bibfnamefont {Vladan}\ \bibnamefont {Vuleti\ifmmode~\acute{c}\else
  \'{c}\fi{}}}, \bibinfo {author} {\bibfnamefont {Hannes}\ \bibnamefont
  {Pichler}}, \ and\ \bibinfo {author} {\bibfnamefont {Mikhail~D.}\
  \bibnamefont {Lukin}},\ }\bibfield  {title} {\enquote {\bibinfo {title}
  {{Parallel Implementation of High-Fidelity Multiqubit Gates with Neutral
  Atoms}},}\ }\href {\doibase 10.1103/PhysRevLett.123.170503} {\bibfield
  {journal} {\bibinfo  {journal} {Phys. Rev. Lett.}\ }\textbf {\bibinfo
  {volume} {123}},\ \bibinfo {pages} {170503} (\bibinfo {year}
  {2019})}\BibitemShut {NoStop}%
\bibitem [{\citenamefont {Khazali}\ and\ \citenamefont
  {M\o{}lmer}(2020)}]{khazali20}%
  \BibitemOpen
  \bibfield  {author} {\bibinfo {author} {\bibfnamefont {Mohammadsadegh}\
  \bibnamefont {Khazali}}\ and\ \bibinfo {author} {\bibfnamefont {Klaus}\
  \bibnamefont {M\o{}lmer}},\ }\bibfield  {title} {\enquote {\bibinfo {title}
  {{Fast Multiqubit Gates by Adiabatic Evolution in Interacting Excited-State
  Manifolds of Rydberg Atoms and Superconducting Circuits}},}\ }\href {\doibase
  10.1103/PhysRevX.10.021054} {\bibfield  {journal} {\bibinfo  {journal} {Phys.
  Rev. X}\ }\textbf {\bibinfo {volume} {10}},\ \bibinfo {pages} {021054}
  (\bibinfo {year} {2020})}\BibitemShut {NoStop}%
\bibitem [{\citenamefont {Rasmussen}\ \emph {et~al.}(2020)\citenamefont
  {Rasmussen}, \citenamefont {Groenland}, \citenamefont {Gerritsma},
  \citenamefont {Schoutens},\ and\ \citenamefont {Zinner}}]{rasmussen20}%
  \BibitemOpen
  \bibfield  {author} {\bibinfo {author} {\bibfnamefont {S.~E.}\ \bibnamefont
  {Rasmussen}}, \bibinfo {author} {\bibfnamefont {K.}~\bibnamefont
  {Groenland}}, \bibinfo {author} {\bibfnamefont {R.}~\bibnamefont
  {Gerritsma}}, \bibinfo {author} {\bibfnamefont {K.}~\bibnamefont
  {Schoutens}}, \ and\ \bibinfo {author} {\bibfnamefont {N.~T.}\ \bibnamefont
  {Zinner}},\ }\bibfield  {title} {\enquote {\bibinfo {title} {{Single-step
  implementation of high-fidelity $n$-bit Toffoli gates}},}\ }\href {\doibase
  10.1103/PhysRevA.101.022308} {\bibfield  {journal} {\bibinfo  {journal}
  {Phys. Rev. A}\ }\textbf {\bibinfo {volume} {101}},\ \bibinfo {pages}
  {022308} (\bibinfo {year} {2020})}\BibitemShut {NoStop}%
\bibitem [{\citenamefont {Li}\ \emph {et~al.}(2021)\citenamefont {Li},
  \citenamefont {Guo}, \citenamefont {Jin}, \citenamefont {Yan}, \citenamefont
  {Liang},\ and\ \citenamefont {Su}}]{li21}%
  \BibitemOpen
  \bibfield  {author} {\bibinfo {author} {\bibfnamefont {Meng}\ \bibnamefont
  {Li}}, \bibinfo {author} {\bibfnamefont {F.-Q.}\ \bibnamefont {Guo}},
  \bibinfo {author} {\bibfnamefont {Z.}~\bibnamefont {Jin}}, \bibinfo {author}
  {\bibfnamefont {L.-L.}\ \bibnamefont {Yan}}, \bibinfo {author} {\bibfnamefont
  {E.-J.}\ \bibnamefont {Liang}}, \ and\ \bibinfo {author} {\bibfnamefont
  {S.-L.}\ \bibnamefont {Su}},\ }\bibfield  {title} {\enquote {\bibinfo {title}
  {{Multiple-qubit controlled unitary quantum gate for Rydberg atoms using
  shortcut to adiabaticity and optimized geometric quantum operations}},}\
  }\href {\doibase 10.1103/PhysRevA.103.062607} {\bibfield  {journal} {\bibinfo
   {journal} {Phys. Rev. A}\ }\textbf {\bibinfo {volume} {103}},\ \bibinfo
  {pages} {062607} (\bibinfo {year} {2021})}\BibitemShut {NoStop}%
\bibitem [{\citenamefont {Young}\ \emph {et~al.}(2021)\citenamefont {Young},
  \citenamefont {Bienias}, \citenamefont {Belyansky}, \citenamefont {Kaufman},\
  and\ \citenamefont {Gorshkov}}]{young21}%
  \BibitemOpen
  \bibfield  {author} {\bibinfo {author} {\bibfnamefont {Jeremy~T.}\
  \bibnamefont {Young}}, \bibinfo {author} {\bibfnamefont {Przemyslaw}\
  \bibnamefont {Bienias}}, \bibinfo {author} {\bibfnamefont {Ron}\ \bibnamefont
  {Belyansky}}, \bibinfo {author} {\bibfnamefont {Adam~M.}\ \bibnamefont
  {Kaufman}}, \ and\ \bibinfo {author} {\bibfnamefont {Alexey~V.}\ \bibnamefont
  {Gorshkov}},\ }\bibfield  {title} {\enquote {\bibinfo {title} {{Asymmetric
  Blockade and Multiqubit Gates via Dipole-Dipole Interactions}},}\ }\href
  {\doibase 10.1103/PhysRevLett.127.120501} {\bibfield  {journal} {\bibinfo
  {journal} {Phys. Rev. Lett.}\ }\textbf {\bibinfo {volume} {127}},\ \bibinfo
  {pages} {120501} (\bibinfo {year} {2021})}\BibitemShut {NoStop}%
\bibitem [{\citenamefont {Pelegr{\'\i}}\ \emph {et~al.}(2022)\citenamefont
  {Pelegr{\'\i}}, \citenamefont {Daley},\ and\ \citenamefont
  {Pritchard}}]{pelegri2022high}%
  \BibitemOpen
  \bibfield  {author} {\bibinfo {author} {\bibfnamefont {G.}~\bibnamefont
  {Pelegr{\'\i}}}, \bibinfo {author} {\bibfnamefont {A.~J.}\ \bibnamefont
  {Daley}}, \ and\ \bibinfo {author} {\bibfnamefont {J.~D.}\ \bibnamefont
  {Pritchard}},\ }\bibfield  {title} {\enquote {\bibinfo {title} {High-fidelity
  multiqubit rydberg gates via two-photon adiabatic rapid passage},}\ }\href
  {\doibase https://doi.org/10.1088/2058-9565/ac823a} {\bibfield  {journal}
  {\bibinfo  {journal} {Quantum Science and Technology}\ }\textbf {\bibinfo
  {volume} {7}},\ \bibinfo {pages} {045020} (\bibinfo {year}
  {2022})}\BibitemShut {NoStop}%
\bibitem [{\citenamefont {Bennett}\ and\ \citenamefont
  {Wiesner}(1992)}]{bennett1992communication}%
  \BibitemOpen
  \bibfield  {author} {\bibinfo {author} {\bibfnamefont {Charles~H.}\
  \bibnamefont {Bennett}}\ and\ \bibinfo {author} {\bibfnamefont {Stephen~J.}\
  \bibnamefont {Wiesner}},\ }\bibfield  {title} {\enquote {\bibinfo {title}
  {Communication via one- and two-particle operators on einstein-podolsky-rosen
  states},}\ }\href {\doibase 10.1103/PhysRevLett.69.2881} {\bibfield
  {journal} {\bibinfo  {journal} {Phys. Rev. Lett.}\ }\textbf {\bibinfo
  {volume} {69}},\ \bibinfo {pages} {2881--2884} (\bibinfo {year}
  {1992})}\BibitemShut {NoStop}%
\bibitem [{\citenamefont {Bennett}\ \emph {et~al.}(1993)\citenamefont
  {Bennett}, \citenamefont {Brassard}, \citenamefont {Cr\'epeau}, \citenamefont
  {Jozsa}, \citenamefont {Peres},\ and\ \citenamefont
  {Wootters}}]{bennett1993teleporting}%
  \BibitemOpen
  \bibfield  {author} {\bibinfo {author} {\bibfnamefont {Charles~H.}\
  \bibnamefont {Bennett}}, \bibinfo {author} {\bibfnamefont {Gilles}\
  \bibnamefont {Brassard}}, \bibinfo {author} {\bibfnamefont {Claude}\
  \bibnamefont {Cr\'epeau}}, \bibinfo {author} {\bibfnamefont {Richard}\
  \bibnamefont {Jozsa}}, \bibinfo {author} {\bibfnamefont {Asher}\ \bibnamefont
  {Peres}}, \ and\ \bibinfo {author} {\bibfnamefont {William~K.}\ \bibnamefont
  {Wootters}},\ }\bibfield  {title} {\enquote {\bibinfo {title} {Teleporting an
  unknown quantum state via dual classical and einstein-podolsky-rosen
  channels},}\ }\href {\doibase 10.1103/PhysRevLett.70.1895} {\bibfield
  {journal} {\bibinfo  {journal} {Phys. Rev. Lett.}\ }\textbf {\bibinfo
  {volume} {70}},\ \bibinfo {pages} {1895--1899} (\bibinfo {year}
  {1993})}\BibitemShut {NoStop}%
\bibitem [{\citenamefont {Nielsen}\ and\ \citenamefont
  {Chuang}(2000)}]{nielsen2000quantum}%
  \BibitemOpen
  \bibfield  {author} {\bibinfo {author} {\bibfnamefont {Michael~A.}\
  \bibnamefont {Nielsen}}\ and\ \bibinfo {author} {\bibfnamefont {Isaac~L.}\
  \bibnamefont {Chuang}},\ }\href@noop {} {\emph {\bibinfo {title} {Quantum
  Computation and Quantum Information}}}\ (\bibinfo  {publisher} {Cambridge
  University Press},\ \bibinfo {year} {2000})\BibitemShut {NoStop}%
\bibitem [{\citenamefont {Walborn}\ \emph {et~al.}(2006)\citenamefont
  {Walborn}, \citenamefont {Souto~Ribeiro}, \citenamefont {Davidovich},
  \citenamefont {Mintert},\ and\ \citenamefont
  {Buchleitner}}]{walborn2006experimental}%
  \BibitemOpen
  \bibfield  {author} {\bibinfo {author} {\bibfnamefont {S.~P.}\ \bibnamefont
  {Walborn}}, \bibinfo {author} {\bibfnamefont {P.~H.}\ \bibnamefont
  {Souto~Ribeiro}}, \bibinfo {author} {\bibfnamefont {L.}~\bibnamefont
  {Davidovich}}, \bibinfo {author} {\bibfnamefont {F.}~\bibnamefont {Mintert}},
  \ and\ \bibinfo {author} {\bibfnamefont {A.}~\bibnamefont {Buchleitner}},\
  }\bibfield  {title} {\enquote {\bibinfo {title} {Experimental determination
  of entanglement with a single measurement},}\ }\href {\doibase
  10.1038/nature04627} {\bibfield  {journal} {\bibinfo  {journal} {Nature}\
  }\textbf {\bibinfo {volume} {440}},\ \bibinfo {pages} {1022–1024} (\bibinfo
  {year} {2006})}\BibitemShut {NoStop}%
\bibitem [{\citenamefont {Walborn}\ \emph {et~al.}(2007)\citenamefont
  {Walborn}, \citenamefont {Ribeiro}, \citenamefont {Davidovich}, \citenamefont
  {Mintert},\ and\ \citenamefont {Buchleitner}}]{walborn2007experimental}%
  \BibitemOpen
  \bibfield  {author} {\bibinfo {author} {\bibfnamefont {S.~P.}\ \bibnamefont
  {Walborn}}, \bibinfo {author} {\bibfnamefont {P.~H.~Souto}\ \bibnamefont
  {Ribeiro}}, \bibinfo {author} {\bibfnamefont {L.}~\bibnamefont {Davidovich}},
  \bibinfo {author} {\bibfnamefont {F.}~\bibnamefont {Mintert}}, \ and\
  \bibinfo {author} {\bibfnamefont {A.}~\bibnamefont {Buchleitner}},\
  }\bibfield  {title} {\enquote {\bibinfo {title} {Experimental determination
  of entanglement by a projective measurement},}\ }\href {\doibase
  10.1103/PhysRevA.75.032338} {\bibfield  {journal} {\bibinfo  {journal} {Phys.
  Rev. A}\ }\textbf {\bibinfo {volume} {75}},\ \bibinfo {pages} {032338}
  (\bibinfo {year} {2007})}\BibitemShut {NoStop}%
\bibitem [{\citenamefont {Haug}\ and\ \citenamefont
  {Kim}(2022)}]{haug2022scalable}%
  \BibitemOpen
  \bibfield  {author} {\bibinfo {author} {\bibfnamefont {Tobias}\ \bibnamefont
  {Haug}}\ and\ \bibinfo {author} {\bibfnamefont {M.~S.}\ \bibnamefont {Kim}},\
  }\bibfield  {title} {\enquote {\bibinfo {title} {Scalable measures of magic
  for quantum computers},}\ }\href {https://arxiv.org/abs/2204.10061}
  {\bibfield  {journal} {\bibinfo  {journal} {arXiv preprint arXiv:2204.10061}\
  } (\bibinfo {year} {2022})}\BibitemShut {NoStop}%
\bibitem [{\citenamefont {Daley}\ \emph {et~al.}(2012)\citenamefont {Daley},
  \citenamefont {Pichler}, \citenamefont {Schachenmayer},\ and\ \citenamefont
  {Zoller}}]{daley2012measuring}%
  \BibitemOpen
  \bibfield  {author} {\bibinfo {author} {\bibfnamefont {A.~J.}\ \bibnamefont
  {Daley}}, \bibinfo {author} {\bibfnamefont {H.}~\bibnamefont {Pichler}},
  \bibinfo {author} {\bibfnamefont {J.}~\bibnamefont {Schachenmayer}}, \ and\
  \bibinfo {author} {\bibfnamefont {P.}~\bibnamefont {Zoller}},\ }\bibfield
  {title} {\enquote {\bibinfo {title} {Measuring entanglement growth in quench
  dynamics of bosons in an optical lattice},}\ }\href {\doibase
  10.1103/PhysRevLett.109.020505} {\bibfield  {journal} {\bibinfo  {journal}
  {Phys. Rev. Lett.}\ }\textbf {\bibinfo {volume} {109}},\ \bibinfo {pages}
  {020505} (\bibinfo {year} {2012})}\BibitemShut {NoStop}%
\bibitem [{\citenamefont {Islam}\ \emph {et~al.}(2015)\citenamefont {Islam},
  \citenamefont {Ma}, \citenamefont {Preiss}, \citenamefont {Eric~Tai},
  \citenamefont {Lukin}, \citenamefont {Rispoli},\ and\ \citenamefont
  {Greiner}}]{islam2015measuring}%
  \BibitemOpen
  \bibfield  {author} {\bibinfo {author} {\bibfnamefont {Rajibul}\ \bibnamefont
  {Islam}}, \bibinfo {author} {\bibfnamefont {Ruichao}\ \bibnamefont {Ma}},
  \bibinfo {author} {\bibfnamefont {Philipp~M.}\ \bibnamefont {Preiss}},
  \bibinfo {author} {\bibfnamefont {M.}~\bibnamefont {Eric~Tai}}, \bibinfo
  {author} {\bibfnamefont {Alexander}\ \bibnamefont {Lukin}}, \bibinfo {author}
  {\bibfnamefont {Matthew}\ \bibnamefont {Rispoli}}, \ and\ \bibinfo {author}
  {\bibfnamefont {Markus}\ \bibnamefont {Greiner}},\ }\bibfield  {title}
  {\enquote {\bibinfo {title} {Measuring entanglement entropy in a quantum
  many-body system},}\ }\href {\doibase 10.1038/nature15750} {\bibfield
  {journal} {\bibinfo  {journal} {Nature}\ }\textbf {\bibinfo {volume} {528}},\
  \bibinfo {pages} {77–83} (\bibinfo {year} {2015})}\BibitemShut {NoStop}%
\bibitem [{\citenamefont {Kaufman}\ \emph {et~al.}(2016)\citenamefont
  {Kaufman}, \citenamefont {Tai}, \citenamefont {Lukin}, \citenamefont
  {Rispoli}, \citenamefont {Schittko}, \citenamefont {Preiss},\ and\
  \citenamefont {Greiner}}]{kaufman2016quantum}%
  \BibitemOpen
  \bibfield  {author} {\bibinfo {author} {\bibfnamefont {Adam~M.}\ \bibnamefont
  {Kaufman}}, \bibinfo {author} {\bibfnamefont {M.~Eric}\ \bibnamefont {Tai}},
  \bibinfo {author} {\bibfnamefont {Alexander}\ \bibnamefont {Lukin}}, \bibinfo
  {author} {\bibfnamefont {Matthew}\ \bibnamefont {Rispoli}}, \bibinfo {author}
  {\bibfnamefont {Robert}\ \bibnamefont {Schittko}}, \bibinfo {author}
  {\bibfnamefont {Philipp~M.}\ \bibnamefont {Preiss}}, \ and\ \bibinfo {author}
  {\bibfnamefont {Markus}\ \bibnamefont {Greiner}},\ }\bibfield  {title}
  {\enquote {\bibinfo {title} {Quantum thermalization through entanglement in
  an isolated many-body system},}\ }\href {\doibase 10.1126/science.aaf6725}
  {\bibfield  {journal} {\bibinfo  {journal} {Science}\ }\textbf {\bibinfo
  {volume} {353}},\ \bibinfo {pages} {794--800} (\bibinfo {year}
  {2016})}\BibitemShut {NoStop}%
\bibitem [{\citenamefont {Bluvstein}\ \emph {et~al.}(2022)\citenamefont
  {Bluvstein}, \citenamefont {Levine}, \citenamefont {Semeghini}, \citenamefont
  {Wang}, \citenamefont {Ebadi}, \citenamefont {Kalinowski}, \citenamefont
  {Keesling}, \citenamefont {Maskara}, \citenamefont {Pichler}, \citenamefont
  {Greiner},\ and\ \citenamefont {et~al.}}]{bluvstein2022quantum}%
  \BibitemOpen
  \bibfield  {author} {\bibinfo {author} {\bibfnamefont {Dolev}\ \bibnamefont
  {Bluvstein}}, \bibinfo {author} {\bibfnamefont {Harry}\ \bibnamefont
  {Levine}}, \bibinfo {author} {\bibfnamefont {Giulia}\ \bibnamefont
  {Semeghini}}, \bibinfo {author} {\bibfnamefont {Tout~T.}\ \bibnamefont
  {Wang}}, \bibinfo {author} {\bibfnamefont {Sepehr}\ \bibnamefont {Ebadi}},
  \bibinfo {author} {\bibfnamefont {Marcin}\ \bibnamefont {Kalinowski}},
  \bibinfo {author} {\bibfnamefont {Alexander}\ \bibnamefont {Keesling}},
  \bibinfo {author} {\bibfnamefont {Nishad}\ \bibnamefont {Maskara}}, \bibinfo
  {author} {\bibfnamefont {Hannes}\ \bibnamefont {Pichler}}, \bibinfo {author}
  {\bibfnamefont {Markus}\ \bibnamefont {Greiner}}, \ and\ \bibinfo {author}
  {\bibnamefont {et~al.}},\ }\bibfield  {title} {\enquote {\bibinfo {title} {A
  quantum processor based on coherent transport of entangled atom arrays},}\
  }\href {\doibase 10.1038/s41586-022-04592-6} {\bibfield  {journal} {\bibinfo
  {journal} {Nature}\ }\textbf {\bibinfo {volume} {604}},\ \bibinfo {pages}
  {451–456} (\bibinfo {year} {2022})}\BibitemShut {NoStop}%
\bibitem [{\citenamefont {Huang}\ \emph {et~al.}(2022)\citenamefont {Huang},
  \citenamefont {Broughton}, \citenamefont {Cotler}, \citenamefont {Chen},
  \citenamefont {Li}, \citenamefont {Mohseni}, \citenamefont {Neven},
  \citenamefont {Babbush}, \citenamefont {Kueng}, \citenamefont {Preskill},\
  and\ \citenamefont {et~al.}}]{huang2022quantum}%
  \BibitemOpen
  \bibfield  {author} {\bibinfo {author} {\bibfnamefont {Hsin-Yuan}\
  \bibnamefont {Huang}}, \bibinfo {author} {\bibfnamefont {Michael}\
  \bibnamefont {Broughton}}, \bibinfo {author} {\bibfnamefont {Jordan}\
  \bibnamefont {Cotler}}, \bibinfo {author} {\bibfnamefont {Sitan}\
  \bibnamefont {Chen}}, \bibinfo {author} {\bibfnamefont {Jerry}\ \bibnamefont
  {Li}}, \bibinfo {author} {\bibfnamefont {Masoud}\ \bibnamefont {Mohseni}},
  \bibinfo {author} {\bibfnamefont {Hartmut}\ \bibnamefont {Neven}}, \bibinfo
  {author} {\bibfnamefont {Ryan}\ \bibnamefont {Babbush}}, \bibinfo {author}
  {\bibfnamefont {Richard}\ \bibnamefont {Kueng}}, \bibinfo {author}
  {\bibfnamefont {John}\ \bibnamefont {Preskill}}, \ and\ \bibinfo {author}
  {\bibnamefont {et~al.}},\ }\bibfield  {title} {\enquote {\bibinfo {title}
  {Quantum advantage in learning from experiments},}\ }\href {\doibase
  10.1126/science.abn7293} {\bibfield  {journal} {\bibinfo  {journal}
  {Science}\ }\textbf {\bibinfo {volume} {376}},\ \bibinfo {pages}
  {1182–1186} (\bibinfo {year} {2022})}\BibitemShut {NoStop}%
\bibitem [{\citenamefont {Mintert}\ \emph {et~al.}(2005)\citenamefont
  {Mintert}, \citenamefont {Ku\ifmmode~\acute{s}\else \'{s}\fi{}},\ and\
  \citenamefont {Buchleitner}}]{mintert2005concurrence}%
  \BibitemOpen
  \bibfield  {author} {\bibinfo {author} {\bibfnamefont {Florian}\ \bibnamefont
  {Mintert}}, \bibinfo {author} {\bibfnamefont {Marek}\ \bibnamefont
  {Ku\ifmmode~\acute{s}\else \'{s}\fi{}}}, \ and\ \bibinfo {author}
  {\bibfnamefont {Andreas}\ \bibnamefont {Buchleitner}},\ }\bibfield  {title}
  {\enquote {\bibinfo {title} {Concurrence of mixed multipartite quantum
  states},}\ }\href {\doibase 10.1103/PhysRevLett.95.260502} {\bibfield
  {journal} {\bibinfo  {journal} {Phys. Rev. Lett.}\ }\textbf {\bibinfo
  {volume} {95}},\ \bibinfo {pages} {260502} (\bibinfo {year}
  {2005})}\BibitemShut {NoStop}%
\bibitem [{\citenamefont {Mintert}\ and\ \citenamefont
  {Buchleitner}(2007)}]{mintert2007observable}%
  \BibitemOpen
  \bibfield  {author} {\bibinfo {author} {\bibfnamefont {Florian}\ \bibnamefont
  {Mintert}}\ and\ \bibinfo {author} {\bibfnamefont {Andreas}\ \bibnamefont
  {Buchleitner}},\ }\bibfield  {title} {\enquote {\bibinfo {title} {Observable
  entanglement measure for mixed quantum states},}\ }\href {\doibase
  10.1103/PhysRevLett.98.140505} {\bibfield  {journal} {\bibinfo  {journal}
  {Phys. Rev. Lett.}\ }\textbf {\bibinfo {volume} {98}},\ \bibinfo {pages}
  {140505} (\bibinfo {year} {2007})}\BibitemShut {NoStop}%
\bibitem [{\citenamefont {Aolita}\ \emph {et~al.}(2008)\citenamefont {Aolita},
  \citenamefont {Buchleitner},\ and\ \citenamefont
  {Mintert}}]{aolita2008scalable}%
  \BibitemOpen
  \bibfield  {author} {\bibinfo {author} {\bibfnamefont {Leandro}\ \bibnamefont
  {Aolita}}, \bibinfo {author} {\bibfnamefont {Andreas}\ \bibnamefont
  {Buchleitner}}, \ and\ \bibinfo {author} {\bibfnamefont {Florian}\
  \bibnamefont {Mintert}},\ }\bibfield  {title} {\enquote {\bibinfo {title}
  {Scalable method to estimate experimentally the entanglement of multipartite
  systems},}\ }\href {\doibase 10.1103/PhysRevA.78.022308} {\bibfield
  {journal} {\bibinfo  {journal} {Phys. Rev. A}\ }\textbf {\bibinfo {volume}
  {78}},\ \bibinfo {pages} {022308} (\bibinfo {year} {2008})}\BibitemShut
  {NoStop}%
\bibitem [{see()}]{seeSupplementary}%
  \BibitemOpen
  \href@noop {} {}\bibinfo {note} {See supplementary material}\BibitemShut
  {NoStop}%
\bibitem [{\citenamefont {Shende}\ and\ \citenamefont
  {Markov}(2008)}]{shende2008on}%
  \BibitemOpen
  \bibfield  {author} {\bibinfo {author} {\bibfnamefont {Vivek~V.}\
  \bibnamefont {Shende}}\ and\ \bibinfo {author} {\bibfnamefont {Igor~L.}\
  \bibnamefont {Markov}},\ }\bibfield  {title} {\enquote {\bibinfo {title} {On
  the cnot-cost of toffoli gates},}\ }\href {https://arxiv.org/abs/0803.2316}
  {\bibfield  {journal} {\bibinfo  {journal} {arXiv preprint arXiv:0803.2316}\
  } (\bibinfo {year} {2008})}\BibitemShut {NoStop}%
\bibitem [{\citenamefont {Garcia-Escartin}\ and\ \citenamefont
  {Chamorro-Posada}(2013)}]{garcia-escartin2013swap}%
  \BibitemOpen
  \bibfield  {author} {\bibinfo {author} {\bibfnamefont {Juan~Carlos}\
  \bibnamefont {Garcia-Escartin}}\ and\ \bibinfo {author} {\bibfnamefont
  {Pedro}\ \bibnamefont {Chamorro-Posada}},\ }\bibfield  {title} {\enquote
  {\bibinfo {title} {swap test and hong-ou-mandel effect are equivalent},}\
  }\href {\doibase 10.1103/PhysRevA.87.052330} {\bibfield  {journal} {\bibinfo
  {journal} {Phys. Rev. A}\ }\textbf {\bibinfo {volume} {87}},\ \bibinfo
  {pages} {052330} (\bibinfo {year} {2013})}\BibitemShut {NoStop}%
\bibitem [{\citenamefont {Meyer}\ and\ \citenamefont
  {Wallach}(2002)}]{meyer2002global}%
  \BibitemOpen
  \bibfield  {author} {\bibinfo {author} {\bibfnamefont {David~A}\ \bibnamefont
  {Meyer}}\ and\ \bibinfo {author} {\bibfnamefont {Nolan~R}\ \bibnamefont
  {Wallach}},\ }\bibfield  {title} {\enquote {\bibinfo {title} {Global
  entanglement in multiparticle systems},}\ }\href {\doibase 10.1063/1.1497700}
  {\bibfield  {journal} {\bibinfo  {journal} {Journal of Mathematical Physics}\
  }\textbf {\bibinfo {volume} {43}},\ \bibinfo {pages} {4273--4278} (\bibinfo
  {year} {2002})}\BibitemShut {NoStop}%
\bibitem [{\citenamefont {Brennen}(2003)}]{brennen2003observable}%
  \BibitemOpen
  \bibfield  {author} {\bibinfo {author} {\bibfnamefont {Gavin~K}\ \bibnamefont
  {Brennen}},\ }\bibfield  {title} {\enquote {\bibinfo {title} {An observable
  measure of entanglement for pure states of multi-qubit systems},}\ }\href
  {https://arxiv.org/abs/quant-ph/0305094} {\bibfield  {journal} {\bibinfo
  {journal} {arXiv preprint quant-ph/0305094}\ } (\bibinfo {year}
  {2003})}\BibitemShut {NoStop}%
\bibitem [{\citenamefont {Carvalho}\ \emph {et~al.}(2004)\citenamefont
  {Carvalho}, \citenamefont {Mintert},\ and\ \citenamefont
  {Buchleitner}}]{carvalho2004decoherence}%
  \BibitemOpen
  \bibfield  {author} {\bibinfo {author} {\bibfnamefont {Andr\'e R.~R.}\
  \bibnamefont {Carvalho}}, \bibinfo {author} {\bibfnamefont {Florian}\
  \bibnamefont {Mintert}}, \ and\ \bibinfo {author} {\bibfnamefont {Andreas}\
  \bibnamefont {Buchleitner}},\ }\bibfield  {title} {\enquote {\bibinfo {title}
  {Decoherence and multipartite entanglement},}\ }\href {\doibase
  10.1103/PhysRevLett.93.230501} {\bibfield  {journal} {\bibinfo  {journal}
  {Phys. Rev. Lett.}\ }\textbf {\bibinfo {volume} {93}},\ \bibinfo {pages}
  {230501} (\bibinfo {year} {2004})}\BibitemShut {NoStop}%
\bibitem [{\citenamefont {Aolita}\ and\ \citenamefont
  {Mintert}(2006)}]{aolita2006measuring}%
  \BibitemOpen
  \bibfield  {author} {\bibinfo {author} {\bibfnamefont {Leandro}\ \bibnamefont
  {Aolita}}\ and\ \bibinfo {author} {\bibfnamefont {Florian}\ \bibnamefont
  {Mintert}},\ }\bibfield  {title} {\enquote {\bibinfo {title} {Measuring
  multipartite concurrence with a single factorizable observable},}\ }\href
  {\doibase 10.1103/PhysRevLett.97.050501} {\bibfield  {journal} {\bibinfo
  {journal} {Phys. Rev. Lett.}\ }\textbf {\bibinfo {volume} {97}},\ \bibinfo
  {pages} {050501} (\bibinfo {year} {2006})}\BibitemShut {NoStop}%
\bibitem [{\citenamefont {Bennett}\ \emph {et~al.}(1996)\citenamefont
  {Bennett}, \citenamefont {DiVincenzo}, \citenamefont {Smolin},\ and\
  \citenamefont {Wootters}}]{bennett1996mixed}%
  \BibitemOpen
  \bibfield  {author} {\bibinfo {author} {\bibfnamefont {Charles~H.}\
  \bibnamefont {Bennett}}, \bibinfo {author} {\bibfnamefont {David~P.}\
  \bibnamefont {DiVincenzo}}, \bibinfo {author} {\bibfnamefont {John~A.}\
  \bibnamefont {Smolin}}, \ and\ \bibinfo {author} {\bibfnamefont {William~K.}\
  \bibnamefont {Wootters}},\ }\bibfield  {title} {\enquote {\bibinfo {title}
  {Mixed-state entanglement and quantum error correction},}\ }\href {\doibase
  10.1103/PhysRevA.54.3824} {\bibfield  {journal} {\bibinfo  {journal} {Phys.
  Rev. A}\ }\textbf {\bibinfo {volume} {54}},\ \bibinfo {pages} {3824--3851}
  (\bibinfo {year} {1996})}\BibitemShut {NoStop}%
\bibitem [{\citenamefont {Uhlmann}(2010)}]{uhlmann2010roofs}%
  \BibitemOpen
  \bibfield  {author} {\bibinfo {author} {\bibfnamefont {Armin}\ \bibnamefont
  {Uhlmann}},\ }\bibfield  {title} {\enquote {\bibinfo {title} {Roofs and
  convexity},}\ }\href {\doibase 10.3390/e12071799} {\bibfield  {journal}
  {\bibinfo  {journal} {Entropy}\ }\textbf {\bibinfo {volume} {12}},\ \bibinfo
  {pages} {1799–1832} (\bibinfo {year} {2010})}\BibitemShut {NoStop}%
\bibitem [{\citenamefont {Wootters}(2001)}]{wootters2001entanglement}%
  \BibitemOpen
  \bibfield  {author} {\bibinfo {author} {\bibfnamefont {William~K.}\
  \bibnamefont {Wootters}},\ }\bibfield  {title} {\enquote {\bibinfo {title}
  {Entanglement of formation and concurrence},}\ }\href {\doibase
  10.26421/qic1.1-3} {\bibfield  {journal} {\bibinfo  {journal} {Quantum
  Information and Computation}\ }\textbf {\bibinfo {volume} {1}},\ \bibinfo
  {pages} {27–44} (\bibinfo {year} {2001})}\BibitemShut {NoStop}%
\bibitem [{\citenamefont {van Enk}\ and\ \citenamefont
  {Beenakker}(2012)}]{vanenk2012measuring}%
  \BibitemOpen
  \bibfield  {author} {\bibinfo {author} {\bibfnamefont {S.~J.}\ \bibnamefont
  {van Enk}}\ and\ \bibinfo {author} {\bibfnamefont {C.~W.~J.}\ \bibnamefont
  {Beenakker}},\ }\bibfield  {title} {\enquote {\bibinfo {title} {Measuring
  $\mathrm{Tr}{\ensuremath{\rho}}^{n}$ on single copies of $\ensuremath{\rho}$
  using random measurements},}\ }\href {\doibase
  10.1103/PhysRevLett.108.110503} {\bibfield  {journal} {\bibinfo  {journal}
  {Phys. Rev. Lett.}\ }\textbf {\bibinfo {volume} {108}},\ \bibinfo {pages}
  {110503} (\bibinfo {year} {2012})}\BibitemShut {NoStop}%
\bibitem [{\citenamefont {Elben}\ \emph {et~al.}(2019)\citenamefont {Elben},
  \citenamefont {Vermersch}, \citenamefont {Roos},\ and\ \citenamefont
  {Zoller}}]{elben2019statistical}%
  \BibitemOpen
  \bibfield  {author} {\bibinfo {author} {\bibfnamefont {A.}~\bibnamefont
  {Elben}}, \bibinfo {author} {\bibfnamefont {B.}~\bibnamefont {Vermersch}},
  \bibinfo {author} {\bibfnamefont {C.~F.}\ \bibnamefont {Roos}}, \ and\
  \bibinfo {author} {\bibfnamefont {P.}~\bibnamefont {Zoller}},\ }\bibfield
  {title} {\enquote {\bibinfo {title} {Statistical correlations between locally
  randomized measurements: A toolbox for probing entanglement in many-body
  quantum states},}\ }\href {\doibase 10.1103/PhysRevA.99.052323} {\bibfield
  {journal} {\bibinfo  {journal} {Phys. Rev. A}\ }\textbf {\bibinfo {volume}
  {99}},\ \bibinfo {pages} {052323} (\bibinfo {year} {2019})}\BibitemShut
  {NoStop}%
\bibitem [{\citenamefont {Brydges}\ \emph {et~al.}(2019)\citenamefont
  {Brydges}, \citenamefont {Elben}, \citenamefont {Jurcevic}, \citenamefont
  {Vermersch}, \citenamefont {Maier}, \citenamefont {Lanyon}, \citenamefont
  {Zoller}, \citenamefont {Blatt},\ and\ \citenamefont
  {Roos}}]{brydges2019probing}%
  \BibitemOpen
  \bibfield  {author} {\bibinfo {author} {\bibfnamefont {Tiff}\ \bibnamefont
  {Brydges}}, \bibinfo {author} {\bibfnamefont {Andreas}\ \bibnamefont
  {Elben}}, \bibinfo {author} {\bibfnamefont {Petar}\ \bibnamefont {Jurcevic}},
  \bibinfo {author} {\bibfnamefont {Benoît}\ \bibnamefont {Vermersch}},
  \bibinfo {author} {\bibfnamefont {Christine}\ \bibnamefont {Maier}}, \bibinfo
  {author} {\bibfnamefont {Ben~P.}\ \bibnamefont {Lanyon}}, \bibinfo {author}
  {\bibfnamefont {Peter}\ \bibnamefont {Zoller}}, \bibinfo {author}
  {\bibfnamefont {Rainer}\ \bibnamefont {Blatt}}, \ and\ \bibinfo {author}
  {\bibfnamefont {Christian~F.}\ \bibnamefont {Roos}},\ }\bibfield  {title}
  {\enquote {\bibinfo {title} {Probing rényi entanglement entropy via
  randomized measurements},}\ }\href {\doibase 10.1126/science.aau4963}
  {\bibfield  {journal} {\bibinfo  {journal} {Science}\ }\textbf {\bibinfo
  {volume} {364}},\ \bibinfo {pages} {260–263} (\bibinfo {year}
  {2019})}\BibitemShut {NoStop}%
\bibitem [{\citenamefont {Ohnemus}\ \emph {et~al.}(2022)\citenamefont
  {Ohnemus}, \citenamefont {Breuer},\ and\ \citenamefont
  {Ketterer}}]{ohnemus2022quantifying}%
  \BibitemOpen
  \bibfield  {author} {\bibinfo {author} {\bibfnamefont {Sophia}\ \bibnamefont
  {Ohnemus}}, \bibinfo {author} {\bibfnamefont {Heinz-Peter}\ \bibnamefont
  {Breuer}}, \ and\ \bibinfo {author} {\bibfnamefont {Andreas}\ \bibnamefont
  {Ketterer}},\ }\bibfield  {title} {\enquote {\bibinfo {title} {Quantifying
  multiparticle entanglement with randomized measurements},}\ }\href
  {https://arxiv.org/abs/2207.13777} {\bibfield  {journal} {\bibinfo  {journal}
  {arXiv preprint arXiv:2207.13777}\ } (\bibinfo {year} {2022})}\BibitemShut
  {NoStop}%
\bibitem [{\citenamefont {Notarnicola}\ \emph {et~al.}(2022)\citenamefont
  {Notarnicola}, \citenamefont {Elben}, \citenamefont {Lahaye}, \citenamefont
  {Browaeys}, \citenamefont {Montangero},\ and\ \citenamefont
  {Vermersch}}]{notarnicola2021randomized}%
  \BibitemOpen
  \bibfield  {author} {\bibinfo {author} {\bibfnamefont {Simone}\ \bibnamefont
  {Notarnicola}}, \bibinfo {author} {\bibfnamefont {Andreas}\ \bibnamefont
  {Elben}}, \bibinfo {author} {\bibfnamefont {Thierry}\ \bibnamefont {Lahaye}},
  \bibinfo {author} {\bibfnamefont {Antoine}\ \bibnamefont {Browaeys}},
  \bibinfo {author} {\bibfnamefont {Simone}\ \bibnamefont {Montangero}}, \ and\
  \bibinfo {author} {\bibfnamefont {Benoit}\ \bibnamefont {Vermersch}},\
  }\bibfield  {title} {\enquote {\bibinfo {title} {A randomized measurement
  toolbox for rydberg quantum technologies},}\ }\href
  {https://arxiv.org/abs/2112.11046} {\bibfield  {journal} {\bibinfo  {journal}
  {arXiv preprint arXiv:2112.11046}\ } (\bibinfo {year} {2022})}\BibitemShut
  {NoStop}%
\bibitem [{\citenamefont {Rath}\ \emph {et~al.}(2021)\citenamefont {Rath},
  \citenamefont {van Bijnen}, \citenamefont {Elben}, \citenamefont {Zoller},\
  and\ \citenamefont {Vermersch}}]{rath2021importance}%
  \BibitemOpen
  \bibfield  {author} {\bibinfo {author} {\bibfnamefont {Aniket}\ \bibnamefont
  {Rath}}, \bibinfo {author} {\bibfnamefont {Rick}\ \bibnamefont {van Bijnen}},
  \bibinfo {author} {\bibfnamefont {Andreas}\ \bibnamefont {Elben}}, \bibinfo
  {author} {\bibfnamefont {Peter}\ \bibnamefont {Zoller}}, \ and\ \bibinfo
  {author} {\bibfnamefont {Beno\^{\i}t}\ \bibnamefont {Vermersch}},\ }\bibfield
   {title} {\enquote {\bibinfo {title} {Importance sampling of randomized
  measurements for probing entanglement},}\ }\href {\doibase
  10.1103/PhysRevLett.127.200503} {\bibfield  {journal} {\bibinfo  {journal}
  {Phys. Rev. Lett.}\ }\textbf {\bibinfo {volume} {127}},\ \bibinfo {pages}
  {200503} (\bibinfo {year} {2021})}\BibitemShut {NoStop}%
\bibitem [{\citenamefont {Huang}\ \emph {et~al.}(2020)\citenamefont {Huang},
  \citenamefont {Kueng},\ and\ \citenamefont {Preskill}}]{huang2020predicting}%
  \BibitemOpen
  \bibfield  {author} {\bibinfo {author} {\bibfnamefont {Hsin-Yuan}\
  \bibnamefont {Huang}}, \bibinfo {author} {\bibfnamefont {Richard}\
  \bibnamefont {Kueng}}, \ and\ \bibinfo {author} {\bibfnamefont {John}\
  \bibnamefont {Preskill}},\ }\bibfield  {title} {\enquote {\bibinfo {title}
  {Predicting many properties of a quantum system from very few
  measurements},}\ }\href {\doibase 10.1038/s41567-020-0932-7} {\bibfield
  {journal} {\bibinfo  {journal} {Nature Physics}\ }\textbf {\bibinfo {volume}
  {16}},\ \bibinfo {pages} {1050–1057} (\bibinfo {year} {2020})}\BibitemShut
  {NoStop}%
\bibitem [{\citenamefont {Elben}\ \emph {et~al.}(2020)\citenamefont {Elben},
  \citenamefont {Kueng}, \citenamefont {Huang}, \citenamefont {van Bijnen},
  \citenamefont {Kokail}, \citenamefont {Dalmonte}, \citenamefont {Calabrese},
  \citenamefont {Kraus}, \citenamefont {Preskill}, \citenamefont {Zoller},\
  and\ \citenamefont {Vermersch}}]{elben2020mixed}%
  \BibitemOpen
  \bibfield  {author} {\bibinfo {author} {\bibfnamefont {Andreas}\ \bibnamefont
  {Elben}}, \bibinfo {author} {\bibfnamefont {Richard}\ \bibnamefont {Kueng}},
  \bibinfo {author} {\bibfnamefont {Hsin-Yuan~(Robert)}\ \bibnamefont {Huang}},
  \bibinfo {author} {\bibfnamefont {Rick}\ \bibnamefont {van Bijnen}}, \bibinfo
  {author} {\bibfnamefont {Christian}\ \bibnamefont {Kokail}}, \bibinfo
  {author} {\bibfnamefont {Marcello}\ \bibnamefont {Dalmonte}}, \bibinfo
  {author} {\bibfnamefont {Pasquale}\ \bibnamefont {Calabrese}}, \bibinfo
  {author} {\bibfnamefont {Barbara}\ \bibnamefont {Kraus}}, \bibinfo {author}
  {\bibfnamefont {John}\ \bibnamefont {Preskill}}, \bibinfo {author}
  {\bibfnamefont {Peter}\ \bibnamefont {Zoller}}, \ and\ \bibinfo {author}
  {\bibfnamefont {Beno\^{\i}t}\ \bibnamefont {Vermersch}},\ }\bibfield  {title}
  {\enquote {\bibinfo {title} {Mixed-state entanglement from local randomized
  measurements},}\ }\href {\doibase 10.1103/PhysRevLett.125.200501} {\bibfield
  {journal} {\bibinfo  {journal} {Phys. Rev. Lett.}\ }\textbf {\bibinfo
  {volume} {125}},\ \bibinfo {pages} {200501} (\bibinfo {year}
  {2020})}\BibitemShut {NoStop}%
\bibitem [{\citenamefont {Jaeger}\ \emph {et~al.}(2003)\citenamefont {Jaeger},
  \citenamefont {Sergienko}, \citenamefont {Saleh},\ and\ \citenamefont
  {Teich}}]{jaeger2003entanglement}%
  \BibitemOpen
  \bibfield  {author} {\bibinfo {author} {\bibfnamefont {Gregg}\ \bibnamefont
  {Jaeger}}, \bibinfo {author} {\bibfnamefont {Alexander~V.}\ \bibnamefont
  {Sergienko}}, \bibinfo {author} {\bibfnamefont {Bahaa E.~A.}\ \bibnamefont
  {Saleh}}, \ and\ \bibinfo {author} {\bibfnamefont {Malvin~C.}\ \bibnamefont
  {Teich}},\ }\bibfield  {title} {\enquote {\bibinfo {title} {Entanglement,
  mixedness, and spin-flip symmetry in multiple-qubit systems},}\ }\href
  {\doibase 10.1103/PhysRevA.68.022318} {\bibfield  {journal} {\bibinfo
  {journal} {Phys. Rev. A}\ }\textbf {\bibinfo {volume} {68}},\ \bibinfo
  {pages} {022318} (\bibinfo {year} {2003})}\BibitemShut {NoStop}%
\bibitem [{\citenamefont {et. al.}(2021)}]{huang2021quantum}%
  \BibitemOpen
  \bibfield  {author} {\bibinfo {author} {\bibfnamefont {Hsin-Yuan~Huang}\
  \bibnamefont {et. al.}},\ }\href {\doibase 10.48550/ARXIV.2112.00778}
  {\enquote {\bibinfo {title}
  {\href{https://arxiv.org/pdf/2112.00778.pdf}{Quantum advantage in learning
  from experiments}},}\ } (\bibinfo {year} {2021})\BibitemShut {NoStop}%
\bibitem [{\citenamefont {Hein}\ \emph {et~al.}(2004)\citenamefont {Hein},
  \citenamefont {Eisert},\ and\ \citenamefont {Briegel}}]{hein2004multiparty}%
  \BibitemOpen
  \bibfield  {author} {\bibinfo {author} {\bibfnamefont {M.}~\bibnamefont
  {Hein}}, \bibinfo {author} {\bibfnamefont {J.}~\bibnamefont {Eisert}}, \ and\
  \bibinfo {author} {\bibfnamefont {H.~J.}\ \bibnamefont {Briegel}},\
  }\bibfield  {title} {\enquote {\bibinfo {title} {Multiparty entanglement in
  graph states},}\ }\href {\doibase 10.1103/PhysRevA.69.062311} {\bibfield
  {journal} {\bibinfo  {journal} {Phys. Rev. A}\ }\textbf {\bibinfo {volume}
  {69}},\ \bibinfo {pages} {062311} (\bibinfo {year} {2004})}\BibitemShut
  {NoStop}%
\end{thebibliography}%
\let\addcontentsline\oldaddcontentsline% Restore \addcontentsline

\onecolumngrid
\setcounter{section}{0}
\setcounter{proposition}{0}
\setcounter{theorem}{0}
\setcounter{corollary}{0}
\setcounter{figure}{0}
\renewcommand{\figurename}{Sup. Fig.}

\newpage
\begin{center}
    \textbf{Supplementary material for \textit{Multipartite entanglement measures via Bell basis measurements}}
\end{center}
In this Supplemental Information, we provide additional details for the manuscript \textit{Multipartite entanglement measures via Bell basis measurements}. First, in Section~\ref{sec:prelims} we present a pedagogical review of elements of quantum information theory necessary to understand our work. We proceed in Section~\ref{sec:proofs} with the proofs and derivations of the main results of our manuscript. Then, in Section \ref{sec:simulations} we provide more details on the simulation data that appear in the manuscript.
\tableofcontents 

\section{Preliminaries}\label{sec:prelims}
We begin by introducing the pure state entanglement measures that appear in the main text. Our main results involve the construction of unbiased estimators of these measures using only the outcomes of Bell basis measurements. For more background on entanglement and various methods of its quantification see Ref. \cite{horodecki2009quantum} and references therein.

\subsection{Multipartite entanglement measures}
\subsubsection{Concentratable entanglements }
Through a detailed numerical investigation of the so-called parallelized controlled SWAP test (see main text for details), the authors in Ref. \cite{foulds2021controlled} conjectured that the outcomes could be used to construct a pure states entanglement monotone. The authors of Ref. \cite{beckey2021computable} proved that the circuit can be used to produce a whole family of pure state entanglement monotones -- dubbed the \textit{concentratable entanglements} (CEs). Before reproducing Def. 1 from the main text, note that throughout this supplementary material, we let $\ket{\psi} \in (\mathbb{C}^2)^{\otimes n}$ be an $n$-qubit pure quantum state. Further, we denote the set of qubit labels within $\ket{\psi}$ as $\SC=\{1, 2, \ldots, n\}$, and $\PC(\SC)$ as its power set (i.e., the set of subsets, with cardinality $|\SC| =2^n$).
\begin{definition}[Concentratable entanglements \cite{beckey2021computable}]
For any non-empty set of qubit labels $s\in \PC(\SC) \setminus \{\emptyset\}$, the Concentratable Entanglement is defined as 
\begin{align}\label{eq:CE}
    \CC_{\ket{\psi}}(s) &= 1 -\frac{1}{2^{|s|}} \sum_{\alpha \in \PC(s)} \tr{\rho_{\alpha}^2},
\end{align}
where the $\rho_{\alpha}$'s are reduced states of $\ket{\psi}\bra{\psi}$ obtained by tracing out subsystems with labels not in $\alpha$. For the trivial subset, we take $\tr{\rho_{\emptyset}^2}:=1$.
\end{definition}
Many well-known measures are recovered as special cases of the CEs. Moreover, as mentioned in the main text, one can compute the CEs using the outcomes of the parallelized c-SWAP test. This, and the other interested properties outlined in Ref. \cite{beckey2021computable}, make the CEs and interest family of entanglement measures to study. 
\subsubsection{Generalized concurrences}
In Ref. \cite{wootters2001entanglement}, Wooters introduced the, now well-known entanglement monotone called the \textit{concurrence}. For pure bipartite quantum states, $\rho_{AB}$, can be compactly expressed as
\begin{align}
    c_2 (\rho_{AB}) &= \sqrt{2(1-\tr{\rho_A^2})},
\end{align}
where we could have equivalently used $\rho_B$ because $\tr{\rho_A^2} =\tr{\rho_B^2}$ for pure states (this follows directly from the Schmidt decomposition \cite{nielsen2000quantum}). By design, $0 \leq C_2(\rho)\leq 1$ with the lower bound being saturated by separable product states and the upper bound being saturated by the Bell states.

There are many ways in which one could generalize Wooters' concurrence to multipartite systems. Ref. \cite{carvalho2004decoherence} explores many different generalizations to Wooters' concurrence, but they focus on the following form, which we will herein refer to as \textit{the} generalized concurrence.
\begin{definition}[Generalized concurrence]\label{def:gen-concurrence}
\begin{align}
    c_n(\ket{\psi}) &= 2^{1-\frac{n}{2}}\sqrt{(2^n-2) - \sum_{\alpha}\tr{\rho_{\alpha}^2} },
\end{align} where the sum is over all $2^n -2$ non-trivial subsets of the $n$-qubit state. That is, they omit the empty set and the full set from the power set. 
\end{definition}
Note that the authors of Ref. \cite{carvalho2004decoherence} claim that Greenberger-Horne-Zeilinger ($GHZ$) states maximize this measure. However, as we will see below, this turns out to be false. Also note that when $s=\SC$, the CE and the generalized concurrence are related by the simple expression
    \begin{align}
        c_n (\ket{\psi}) &= 2 \sqrt{\mathcal{C}_{\ket{\psi}}}.
    \end{align}
   To see this, observe
    \begin{align}
         c_n (\ket{\psi}) &= 2^{1-\frac{n}{2}} \sqrt{\left((2^n - 2) \braket{\psi}{\psi}^2 - \sum_{i} \tr{\rho_i^2}\right)},\\
        c_n^2(\ket{\psi}) &= (2^{1-\frac{n}{2}})^2 \left((2^n - 2) - \sum_{i} \tr{\rho_i^2}\right),\\
        &= \frac{4}{2^n}\left((2^n - 2) - \sum_{i} \tr{\rho_i^2}\right),\\
        &= 4- \frac{8}{2^n} - \frac{4}{2^n}\sum_{i} \tr{\rho_i^2},\\
        &= 4 \left(1 - \frac{1}{2^n}\left( 2 + \sum_{i} \tr{\rho_i^2}\right)\right),\\
        &= 4 \left(1 - \frac{1}{2^n}\left( \tr{\rho_{\emptyset}^2} + \tr{\rho^2} + \sum_{i} \tr{\rho_i^2}\right)\right),\\
       &= 4\left(1-\frac{1}{2^n} \sum_{\alpha \in \mathcal{P}(\mathcal{S})} \tr{\rho_{\alpha}^2}\right),\\
        \implies c_n (\ket{\psi}) &= 2 \sqrt{\mathcal{C}_{\ket{\psi}}(\mathcal{S})},
    \end{align}
    as desired.

\subsubsection{$n$-tangle}
The $n$-tangle is the well-studied pure state entanglement monotone \cite{wong2001potential} defined as follows.
\begin{definition}[$n$-tangle] Let $\ket{\psi} \in (\mathbb{C}^2)^{\otimes n}$. The $n$-tangle is defined as 
\begin{align}
    \tau_{(n)} &= |\braket{\psi}{\tilde{\psi}}|^2,
\end{align}
where $\ket{\tilde{\psi}} := \sigma_2^{\otimes n} \ket{\psi^{*}}$ and the $``*"$ denotes complex conjugation. 
\end{definition}

It was shown in Ref. \cite{jaeger2003entanglement} that the following entanglement measure is equivalent to the $n$-tangle for pure state inputs
\begin{align}
    S_{(n)}^2 := \frac{1}{2^n} \left((S_{0\dots 0})^2 - \sum_{k=1}^n \sum_{i_k}^3 (S_{0\dots i_k \dots 0})^2 + \sum_{k,l=1}^n \sum_{i_k, i_l =1}^3 (S_{0\dots i_k \dots i_l \dots 0})^2 - \dotsm + (-1)^n \sum_{i_1, \dots, i_n}^3 (S_{i_1 \dots i_n})^2\right),
\end{align}
where $S_{i_1,\dots,i_n}=\tr{\rho \sigma_{i_1}\otimes \dotsm \otimes \sigma_{i_n}}$ for $i_1,\dots,i_n \in \{ 0,1,2,3 \}$ are the so-called $n$-qubit Stokes parameters. We use the fact that $S^2_{(n)}=\tau_{(n)}$ for pure state inputs below. Now that we have introduced the entanglement measures we are interested in, we can proceed to some crucial facts regarding the SWAP operator. 

\subsection{Representations and properties of the n-qubit SWAP operator}
\subsubsection{Single-qubit SWAP operator}\label{sec:SWAP-operator}
Consider a Hilbert space of the form $\cal H \otimes \cal H$. Let $\{\ket{j}\}$ be an orthonormal basis of $\cal H$, so that $\mathcal{B} = \{\ket{j}\ket{j'}\}$ is an orthonormal product basis of $\cal H \otimes \cal H$. The single-qubit SWAP operator $\mathbb{F}:\cal H \otimes \cal H \rightarrow \cal H \otimes \cal H$ is defined by its action on the elements of $\mathcal{B}$:
\begin{align}\label{eq:SWAP}
    \mathbb{F}\,\ket{j}\ket{j'}=\ket{j'}\ket{j}\,\,\quad \forall\, \quad \ket{j}\ket{j'}\in \mathcal{B}. 
\end{align}

Next, recall that the Bell basis contains the following elements
\begin{align}
    \ket{\Phi^+} &=\frac{1}{\sqrt{2}}(\ket{0}\ket{0} + \ket{1}\ket{1}), \quad 
    \ket{\Psi^+} = \frac{1}{\sqrt{2}}(\ket{0}\ket{1} + \ket{1}\ket{0})
    ,\\
    \ket{\Phi^-} &= \frac{1}{\sqrt{2}}(\ket{0}\ket{0} - \ket{1}\ket{1}), \quad
    \ket{\Psi^-} = \frac{1}{\sqrt{2}}(\ket{0}\ket{1} - \ket{1}\ket{0}).
\end{align}
One way they can be obtained from the computational basis vectors is by applying a Hadamard and then a CNOT. Explicitly, this yields 
\begin{align}
    CNOT (H \otimes \mathbb{I}) \ket{0}\ket{0} &=CNOT\left( \frac{1}{\sqrt{2}}(\ket{0}\ket{0} + \ket{1}\ket{0})\right) =\frac{1}{\sqrt{2}}(\ket{0}\ket{0} + \ket{1}\ket{1}) = \ket{\Phi^+},\\
    CNOT (H \otimes \mathbb{I}) \ket{0}\ket{1} &=CNOT\left( \frac{1}{\sqrt{2}}(\ket{0}\ket{1} + \ket{1}\ket{1})\right) =\frac{1}{\sqrt{2}}(\ket{0}\ket{1} + \ket{1}\ket{0}) = \ket{\Psi^+} ,\\
    CNOT (H \otimes \mathbb{I}) \ket{1}\ket{0} &=CNOT\left( \frac{1}{\sqrt{2}}(\ket{0}\ket{0} - \ket{1}\ket{0})\right) =\frac{1}{\sqrt{2}}(\ket{0}\ket{0} - \ket{1}\ket{1}) = \ket{\Phi^-},\\
    CNOT (H \otimes \mathbb{I}) \ket{1}\ket{1} &=CNOT\left( \frac{1}{\sqrt{2}}(\ket{0}\ket{1} - \ket{1}\ket{1})\right) =\frac{1}{\sqrt{2}}(\ket{0}\ket{1} - \ket{1}\ket{0}) = \ket{\Psi^-} .
\end{align}
We note here the importance of the above relationships. Because the Bell basis is obtained via just a Hadamard and a CNOT, to carry out a Bell basis measurement, one can simply apply those gates to the relevant pairs of qubits and then do a standard computational basis measurement (see Sup. Fig. \ref{fig:bell_circuit}). This is experimentally feasible because most, if not all, gate-model quantum computers being built today are endowed with the ability to perform these operations. Thus, this method could be implemented on most hardware being built today. 

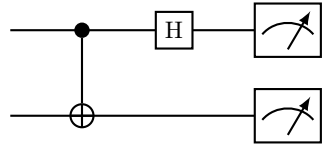
\begin{figure}[h!]\label{fig:bell_circuit}
    \centering
    \begin{tikzpicture}[thick]
        \tikzset{operator/.style = {draw,fill=white,minimum size=1.5em}, operator2/.style = {draw,fill=white,minimum height=2cm, minimum width=1cm}, phase/.style = {draw,fill,shape=circle,minimum size=5pt,inner sep=0pt}, phase2/.style = {draw,fill=black,shape=circle,minimum size=5pt,inner sep=0pt}, phase3/.style = {draw,fill,shape=circle,minimum size=5pt,inner sep=0pt}, surround/.style = {fill=black!10,thick,draw=black,rounded corners=2mm}, cross/.style={path picture={ \draw[thick,black](path picture bounding box.north) -- (path picture bounding box.south) (path picture bounding box.west) -- (path picture bounding box.east);}}, circlewc/.style={draw,circle,cross,minimum width=0.3 cm}, cross2/.style={path picture={ \draw[thick](path picture bounding box.north) -- (path picture bounding box.south) (path picture bounding box.west) -- (path picture bounding box.east);}}, circlewc2/.style={draw,color=black,circle,cross2,minimum width=0.3 cm}, cross3/.style={path picture={ \draw[thick](path picture bounding box.north) -- (path picture bounding box.south) (path picture bounding box.west) -- (path picture bounding box.east); }}, circlewc3/.style={draw,color=black,circle,cross3,minimum width=0.3 cm}, meter/.style= {draw, fill=white, inner sep=7, rectangle, font=\vphantom{A}, minimum width=25, line width=.8, path picture={\draw[black] ([shift={(.1,.3)}]path picture bounding box.south west) to[bend left=50] ([shift={(-.1,.3)}]path picture bounding box.south east);\draw[black,-latex] ([shift={(0,.1)}]path picture bounding box.south) -- ([shift={(.3,-.1)}]path picture bounding box.north);}},} \matrix[row sep=0.4cm, column sep=0.8cm] (circuit) {\node (q1) {};
        &\node[phase2] (P11) {}; 
        &\node[operator] (H12) {H};
        &\node[meter] (meter) {};
        \coordinate (end1);     
        \\ \node (q2) {}; 
        &\node[circlewc2] (P21) {}; 
        & &[-1.2cm] \node[meter] (meter) {};
        \coordinate (end2); \\};
        \begin{pgfonlayer}{background} \draw[thick] (q1) -- (end1) (q2) -- (end2); \draw[thick] (P11) -- (P21) %(P13) -- (P23)
        ;
        %\draw[decorate, decoration = {brace, mirror, amplitude=5pt}] (3.2,-1.5) --  (3.2,1.3) node[pos=0.5,right=5pt,black]{$\otimes n, \otimes M$};%
        \end{pgfonlayer}
    \end{tikzpicture}
    \caption{\textbf{Bell basis measurement.} By applying a CNOT gate, followed by a Hadamard gate, one can convert a computational basis measurement to a Bell basis measurement.}
\end{figure}

Now we come to the reason the Bell basis measurement scheme can simulate the parallelized c-SWAP circuit: the Bell basis vectors are the eigenstates of the SWAP operator
\begin{align}
    \mathbb{F} \ket{\Phi^+} = \ket{\Phi^+}, \quad \mathbb{F} \ket{\Phi^-} = \ket{\Phi^-}, \quad 
    \mathbb{F} \ket{\Psi^+} = \ket{\Psi^+}, \quad \text{and} \quad
    \mathbb{F} \ket{\Psi^-} &= -\ket{\Psi^-}. 
\end{align}
Those with a positive eigenvalue are called \textit{triplet} states. The remaining state is called the \textit{singlet}. The eigenspace spanned by the triplet states is called the \textit{symmetric subspace} and the orthogonal complement, spanned by the singlet state, the \textit{antisymmetric subspace}.  The projectors onto these subspaces are given as
\begin{align}
    \Pi_+ &:= \ket{\Phi^+}\bra{\Phi^+} +\ket{\Phi^-}\bra{\Phi^-} + \ket{\Psi^+}\bra{\Psi^+}, \\
    \Pi_- &:= \ket{\Psi^-}\bra{\Psi^-}.
\end{align}
The swap operator can thus be represented as the difference of these two operators
\begin{align}
    \mathbb{F} &= \Pi_+ - \Pi_-.
\end{align}
Because a basis is complete, by definition, the sum of these projectors is the identity operator on the space of two qubits
\begin{align}
    \mathbb{I}\otimes \mathbb{I} &= \ket{\Phi^+}\bra{\Phi^+} +\ket{\Phi^-}\bra{\Phi^-} + \ket{\Psi^+}\bra{\Psi^+} + \ket{\Psi^-}\bra{\Psi^-}.
\end{align}
From these two expressions, it follows that the projectors can be expressed as
\begin{align}
    \Pi_+ = \frac{\mathbb{I}\otimes \mathbb{I} + \mathbb{F}}{2} \quad \text{ and } \quad \Pi_- = \frac{\mathbb{I}\otimes \mathbb{I} - \mathbb{F}}{2}.
\end{align}
As we will see with the $n$-qubit SWAP operator, it is usually necessary to specify which states are being acted upon by the SWAP operator. We thus denote the SWAP between the $k$-th test and copy qubits as $\mathbb{F}_k$. Although the Bell basis decomposition of the SWAP operator will be the one that is primarily used, another useful decomposition is given in the Pauli basis as

\begin{align}\label{eq:SWAPfromPauli}
    \mathbb{F}_k &= \frac{1}{2}\left(\sigma_{0_k} + \sigma_{1_k} + \sigma_{2_k} + \sigma_{3_k}\right),
\end{align}
where $\sigma_{i_k}$ represents the $i$-th Pauli matrix on both the test and copy qubits. Explicitly,

\begin{align}
    \sigma_{0_k} &:= \mathbb{I}_k \otimes \mathbb{I}_{k'}, \quad \sigma_{1_k} := X_k \otimes X_{k'}, \quad \sigma_{2_k} := Y_k \otimes Y_{k'}, \quad  \sigma_{3_k} := Z_k \otimes Z_{k'}.
\end{align}
We will use $\{\mathbb{I},X,Y,Z\}$ to represent the Pauli matrices unless $\sigma_{i_k}$ allows for more compact notation.

\subsubsection{The $n$-qubit SWAP operator}\label{sec:n-qubitSWAP}
To extend to the multi-qubit regime, we let the test and copy Hilbert spaces have a tensor product structure themselves. That is, let 
\begin{align}
    \mathcal{H} &= \bigotimes_{j=1}^n \mathcal{H}_j.
\end{align}
Further, denote the computational basis of this $n$-qubit space as
\begin{align}
    \mathcal{B} &= \{\ket{\boldsymbol{j}} = \bigotimes_{k=1}^n \ket{j_k}\},
\end{align}
where $j_k\in\{0,1\}$. Because we have a test and copy state, our full space will be $\mathcal{H}\otimes \mathcal{H}$ with basis $\{\ket{\boldsymbol{j}}\ket{\boldsymbol{j}'}\}$. The $n$-qubit SWAP operator acts on this basis as
\begin{align}
   \mathbb{F} \ket{\boldsymbol{j}}\ket{\boldsymbol{j}'}&= \ket{\boldsymbol{j}'}\ket{\boldsymbol{j}}
\end{align}
The crucial observation needed to handle $n$-qubits states is that the $n$-qubit SWAP operator can be written as the $n$-fold tensor product of single qubit SWAP operators
\begin{align}
    \label{eq:n-qubitSWAP}
    \mathbb{F}&= \bigotimes_{j=1}^n \mathbb{F}_j,
\end{align}
where $F_j : \mathcal{H}_j \otimes \mathcal{H}_{j'} \rightarrow \mathcal{H}_j \otimes \mathcal{H}_{j'}$ is the single qubit SWAP operator acting on the $j$-th qubits of the test and copy system. When it should be clear by context, we will simply denote the $n$-qubit SWAP operator as $\mathbb{F}$. We now state an important lemma upon which most methods of purity estimation rely.
\begin{lemma}[The swap ``trick''] \label{lemma:swap-trick}
For an $n$-qubit state $\rho$, the following equality holds
\begin{align}
    \tr{\mathbb{F}\rho^{\otimes 2}} &= \tr{\rho^2}.
\end{align}
This is commonly referred to as the swap ``trick.''
\end{lemma}
\begin{proof}
Let $\rho,\sigma$ be two $n$-qubit quantum states with spectral decompositions given as
\begin{align}
    \rho=\sum_{\boldsymbol{i}} a_{\boldsymbol{i}} \ket{\boldsymbol{i}}\bra{\boldsymbol{i}} \quad \text{ and } \quad \sigma=\sum_{\boldsymbol{j}} b_{\boldsymbol{j}} \ket{\boldsymbol{j}}\bra{\boldsymbol{j}}.
\end{align}
This allows us to write
\begin{align}
    \tr{\mathbb{F} \rho \otimes \sigma} &= \tr{\mathbb{F} \sum_{\boldsymbol{i},\boldsymbol{j}} a_{\boldsymbol{i}} b_{\boldsymbol{j}} \ket{\boldsymbol{i}}\bra{\boldsymbol{i}} \otimes \ket{\boldsymbol{j}}\bra{\boldsymbol{j}}},\\  
    &= \tr{\sum_{\boldsymbol{i},\boldsymbol{j}} a_{\boldsymbol{i}} b_{\boldsymbol{j}} \mathbb{F}  \ket{\boldsymbol{i}}\bra{\boldsymbol{i}} \otimes \ket{\boldsymbol{j}}\bra{\boldsymbol{j}}},\\  
    &= \tr{\sum_{\boldsymbol{i},\boldsymbol{j}} a_{\boldsymbol{i}} b_{\boldsymbol{j}}   \ket{\boldsymbol{j}}\bra{\boldsymbol{i}} \otimes \ket{\boldsymbol{i}}\bra{\boldsymbol{j}}},\\  
    &= \sum_{\boldsymbol{k}} \bra{\boldsymbol{k}}  \bra{\boldsymbol{k}} \left(\sum_{\boldsymbol{i},\boldsymbol{j}} a_{\boldsymbol{i}} b_{\boldsymbol{j}}   \ket{\boldsymbol{j}}\bra{\boldsymbol{i}} \otimes \ket{\boldsymbol{i}}\bra{\boldsymbol{j}}\right)\ket{\boldsymbol{k}}  \ket{\boldsymbol{k}} ,\\ 
    &= \sum_{\boldsymbol{k}} a_{\boldsymbol{k}} b_{\boldsymbol{k}},\\
    &= \tr{\rho \sigma}.
\end{align}
Letting $\sigma = \rho$ completes the proof.
\end{proof}
Finally, we introduce some results from classical statistics that will be utilized throughout to obtain confidence intervals for our estimators. 
\subsection{Results from classical statistics}
\subsubsection{Hoeffding's inequality}
Hoeffding's inequality is a concentration inequality that applies to independent random variables. We will state it without proof as it is a standard result proven in most mathematical statistics textbooks.

\begin{fact}[Hoeffding's inequality] \label{fact:hoeffding}
Let $X_1, \dots, X_M$ be independent random variables such that $a_i \leq X_i \leq b_i,$ and $\E{X_i}=\mu.$ Then, for any $\epsilon > 0$,
\begin{align}
    \label{eq:Hoeffding-mean}
    \text{Prob}\left( \left|\frac{1}{M}\sum_{i=1}^M X_i - \mu \right| < \epsilon \right) \geq 1- 2\exp{\left(- \frac{2M^2 \epsilon^2}{\sum_{i=1}^M (b_i - a_i)^2}\right)}.
\end{align}
\end{fact}

\subsubsection{Clopper-Pearson Confidence Intervals}
Clopper-Pearson confidence intervals apply to Bernoulli random variables (i.e. random variables that take only two possible values). Consider a series of $M$ Bernoulli trials in which we measure a binary variable $X\in \{0,1\}$ such that $P(X=1)=p, P(X=0)=1-p$. If we keep a register of the outcomes $\{X_1,...,X_M\}$, the probability of obtaining the result $X=1$ $k$ times is given by the binomial distribution
\begin{equation}
P(k)={{M}\choose{k}}p^k(1-p)^{M-k}.
\label{eq:Binomial_pmf}
\end{equation}
We are interested in asking the reverse question, i.e., given that we have obtained the result $X=1$ $k$ times in $M$ trials, what is the underlying value of $p$? According to the Maximum Likelihood Estimation (MLE) procedure, we can find the most likely value $\tilde{p}$ by maximising Eq.~\eqref{eq:Binomial_pmf} with respect to $p$ keeping $k$ and $M$ fixed. Doing this, we find
\begin{equation}
\tilde{p}=\frac{k}{M}.
\end{equation}
The Clopper-Pearson confidence intervals provide bounds for how accurate this estimation of the binomial parameter $p$ is. The upper and lower limits $p_U$, $p_L$ are defined to incorporate all values of $p$ that are included with a probability greater than a threshold $\delta$, which defines a $100\times(1-\delta)\%$ confidence interval.  For $k$ observations in $n$ trials, these bounds are found by solving numerically the following equations
 \begin{align}
&\sum_{i=0}^k {{M}\choose{i}}p_U(k)^i(1-p_U(k))^{M-i}=\delta/2\label{eq:CPU},\\
&\sum_{i=k}^M {{M}\choose{i}}p_L(k)^i(1-p_L(k))^{M-i}=\delta/2.\label{eq:CPL}
\end{align}
These results from classical statistics will be used to determine how many measurement repetitions are needed to obtain $\epsilon$-close estimates of quantities herein. Before proceeding to the proofs of the main results, we will review some relevant results from the literature. While not new, these results provide a gentler introduction to the methods used in our main results. As such, we include detailed proofs for the readers convenience. 
\subsection{Estimating subsystem purities from Bell basis measurements}
One of the first examples of ancilla-free purity estimation was carried out in \cite{daley2012measuring}. It was also discussed in a very interesting paper demonstrating the connections between the Hong-ou-Mandel effect and the SWAP test~\cite{garcia-escartin2013swap}. Since then, other groups have used these methods in cutting-edge experiments \cite{islam2015measuring,kaufman2016quantum,bluvstein2022quantum}. For completeness, and to motivate our extensions, we describe how these methods of ancilla-free purity estimation work in detail.
The punchline of our work is that much can be learned from post-processing Bell basis measurement data in interesting ways. We begin by showing how to estimate the purity of a qubit using two copies of the qubit and Bell basis measurements. We will adopt the notation of Ref.\,\cite{huang2021quantum} herein. 

Suppose we carry out $M$ rounds of Bell basis measurements, as shown in Fig. 2 of the main text. Then for round $m \in \{1,\dots,M\}$, for every $k \in \{1,\dots,n\}$ we will perform a Bell basis measurement on the $k$-th test and copy qubit in $\rho$. This measurement yields, with some probability, one of the Bell basis projectors 
\begin{align}
    B_k^{(m)} \in \left\{ \ket{\Phi^+}\bra{\Phi^+},\ket{\Phi^-}\bra{\Phi^-}, \ket{\Psi^+}\bra{\Psi^+},\ket{\Psi^-}\bra{\Psi^-}\right\}.
\end{align}
Throughout, you should think of $B_k^{(m)}$ as a random variable that takes values in the Bell basis. For each of the $M$ rounds, we efficiently store the qubit label, $k$, and the corresponding measurement outcome $B_k^{(m)}$. This requires $\mathcal{O}(nM)$ classical bits. We can then enter the prediction phase. First, we consider the simplest case: estimating the purity of a single qubit.

\subsubsection{Single qubit purity estimation}
Let us consider estimating the purity of a single qubit $\rho \in \mathbb{C}^2$ using two copies of $\rho$. That is, let the state to be measured be $\rho \otimes \rho \in \mathcal{H} \otimes \mathcal{H}$ with dim $\mathcal{H}=2$. Measuring in the Bell basis will project $\rho \otimes \rho$ into one of the four Bell states. As explained in Sec.\,\ref{sec:SWAP-operator}, each of the Bell states is an eigenstate of the SWAP operator, $\mathbb{F}$, with eigenvalue $\pm 1$. The probability that we project into a state with a $+1$ eigenvalue is
\begin{align}
    \text{Prob}(+1) &= \tr{ \Pi_+ \rho \otimes \rho},\\
    &= \tr{  \left( \frac{\mathbb{I}\otimes \mathbb{I} + \mathbb{F}}{2}\right)\rho \otimes \rho},\\
    &= \frac{1}{2} \left(1 + \tr{\rho^2}\right),
\end{align}
where we have used Lemma \ref{lemma:swap-trick} and the fact that quantum states have unit trace. Similarly for the $-1$ eigenvalue, we find $\text{Prob}(-1) = \frac{1}{2} \left(1 - \tr{\rho^2}\right)$. Let the purity be denoted $\gamma := \tr{\rho^2}$. Then, we can construct an estimator for the purity based on the sample average of Bell basis measurement outcomes as
\begin{align}
    \hat{\gamma} &= \frac{1}{M} \sum_{m=1}^M \tr{\mathbb{F} B^{(m)}},
\end{align}
where we have suppressed the subscript on $B^{(m)}_k$ because we are dealing with a single qubit state. We say that this estimator is unbiased if $\mathbb{E}[\hat{\gamma}]=\gamma$. In the case of a single qubit, we can show this easily  
\begin{align}
    \E{\hat{\gamma}} &= \E{\frac{1}{M} \sum_{m=1}^M \tr{\mathbb{F} B^{(m)}}},\\
    &= \frac{1}{M}\sum_{m=1}^M \E{\tr{\mathbb{F} B^{(m)}}},\\
    &=\E{\tr{\mathbb{F} B^{(m)}}},\\
    &= (+1)\cdot\text{Prob}(+1) + (-1)\cdot\text{Prob}(-1),\\
    &= \frac{1}{2} \left(1 + \tr{\rho^2}\right) - \frac{1}{2} \left(1 - \tr{\rho^2}\right),\\
    &= \tr{\rho^2},\\
    \E{\hat{\gamma}}&= \gamma.
\end{align}
Thus, we have an unbiased estimator of the purity of a single qubit. It takes just a little bit of work to extend this to the $n$-qubit case. 
\subsubsection{$n$-qubit subsystem purity estimation}
We now generalize the results above to handle the estimation of $\tr{\rho_{\alpha}^2}$ for all $\alpha$. That is, for reduced states of dim$\rho_{\alpha} \in \{2,4,\dots,2^n\}$. To do this, we let the test and copy Hilbert spaces have a tensor product structure themselves, as outline in Sec. \ref{sec:n-qubitSWAP}. Next, recall that the eigenstates of the single qubit SWAP operator are the Bell states, with eigenstates $\pm 1$. Thus, because the eigenstates of the $n$-qubit SWAP operator are the $n$-fold tensor products of Bell states, they must also have eigenvalue $\pm 1$. It follows that
\begin{align}
    \mathbb{F} \bigotimes_{k=1}^n B_k^{(m)} &= \pm 1 \bigotimes_{k=1}^n B_k^{(m)} \quad \implies \tr{\mathbb{F} \bigotimes_{k=1}^n B_k^{(m)}}=\pm 1.
\end{align}
This allows us to write
\begin{align}
     \pm 1 &=\tr{\mathbb{F} \bigotimes_{k=1}^n B_k^{(m)}}, \\
     &= \tr{\bigotimes_{i=1}^n \mathbb{F}_i \bigotimes_{k=1}^n B_k^{(m)}},\\
    &= \tr{\bigotimes_{k=1}^n \mathbb{F}_k B_k^{(m)}},\\
    \pm 1 &= \prod_{k=1}^n \tr{\mathbb{F}_k B_k^{(m)}}.
\end{align}
Letting the product run from $1$ to $n$, one would be able to construct an estimator of the full purity of an $n$-qubit state of interest. However, letting the product only run over a subset of qubit labels allows one to construct estimators for any subsystem purity. To see this, first let $\gamma_{\alpha} := \tr{\rho_{\alpha}^2}$. Then, remembering that $\alpha$ denotes the set of qubit labels we are interest in, we can construct estimators for these purities as
\begin{align}\label{eq:subsytem-purity}
    \hat{\gamma}_{\alpha} &= \frac{1}{M} \sum_{m=1}^M \prod_{k \in \alpha }\tr{\mathbb{F}_k B_k^{(m)}}.
\end{align}
To see that this is an unbiased estimator of subsystem purity, we consider the expectation value with respect to Bell basis measurement outcomes of this quantity 
\begin{align}
    \E{ \hat{\gamma}_{\alpha}} &=\E{\frac{1}{M} \sum_{m=1}^M \prod_{k \in \alpha }\tr{\mathbb{F}_k B_k^{(m)}}},\\
    &= \frac{1}{M} \sum_{m=1}^M \E{ \prod_{k \in \alpha }\tr{\mathbb{F}_k B_k^{(m)}}},\\
    &= \E{ \prod_{k \in \alpha }\tr{\mathbb{F}_k B_k^{(m)}}},\\
    &= (+1)\cdot \text{Prob}(+) + (-1)\cdot \text{Prob}(-),\\
    &= \tr{\Pi_+ \rho_{\alpha} \otimes \rho_{\alpha}} - \tr{\Pi_- \rho_{\alpha} \otimes \rho_{\alpha}}, \\
    &=\tr{\frac{\mathbb{I} + \mathbb{F}}{2} \rho_{\alpha} \otimes \rho_{\alpha}} - \tr{\frac{\mathbb{I} - \mathbb{F}}{2} \rho_{\alpha} \otimes \rho_{\alpha}},\\
    &= \tr{\mathbb{F} \rho_{\alpha} \otimes \rho_{\alpha}},\\
    \E{ \hat{\gamma}_{\alpha}}&= \tr{\rho_{\alpha }^2}, 
\end{align}
as desired. Note that here $\mathbb{F}$ only acts on the test and copy qubits labelled by the set $\alpha$. It then follows from Hoeffding's inequality that given $\Theta{(\log{(1/\delta)}/\epsilon^2)}$ measurements, we have
\begin{align}
     \text{Prob}\left( \left|\hat{\gamma}_{\alpha}-\gamma_{\alpha} \right| < \epsilon \right) &\geq 1- 2\exp{\left(\frac{-N \epsilon^2}{2}\right)}.
\end{align}
We note, however, that because subsystem purities can be as small as $\frac{1}{2^{|\alpha|}}$, one must set $\epsilon \sim \frac{1}{2^{|\alpha|}}$. Thus, the number of measurements required to obtain an $\epsilon$-close approximation of subsystem purity scales with the square of the subsystem dimension. That is, $ M \sim \Theta ( \log{(1/\delta)} \cdot  4^{|\alpha|})$. 

With these fundamentals and previously known results in mind, we are can proceed to the proofs of the main results in the text.

\section{Proofs of main results}\label{sec:proofs}

\subsection{Unbiased estimation of CE via Bell basis measurements}\label{sec:CE-from-bell}
We want to construct an estimator, $\hat{\mathcal{C}}_{\ket{\psi}}(s)$, that depends only on the Bell basis measurement outcomes and whose expectation value is the concentratable entanglement
\begin{align}
    \E{\hat{\mathcal{C}}_{\ket{\psi}}(s)} &= 1 - \frac{1}{2^{|s|}}\sum_{\alpha \in \mathcal{P}(s)} \text{tr}[\rho_{\alpha}^2]. 
\end{align}
Because the Bell states are eigenstates of the SWAP operator, we know $\tr{\mathbb{F}_k B^{(m)}_k} = \pm 1$. Thus, the outcome of measuring the $k$-th test and copy qubit in the Bell basis will be $\pm 1$. With this in mind, we can state the main theorem from the text.
\begin{theorem}\label{thm:CE-from-Bell}
The quantity 
\begin{align}
    \hat{\mathcal{C}}_{\ket{\psi}}(s) &= 1- \frac{1}{M}\sum_{m=1}^M \prod_{k\in s} \left(\frac{1+\tr{\mathbb{F}_k B^{(m)}_k}}{2} \right),
\end{align}
is an unbiased estimator of the concentratable entanglement. That is,
\begin{align}
    \E{\hat{\mathcal{C}}_{\ket{\psi}}(s)} &= 1 - \frac{1}{2^{|s|}}\sum_{\alpha \in \mathcal{P}(s)} \text{tr}[\rho_{\alpha}^2],
\end{align}
where the expectation value is with respect to the probability distribution induced by the Bell basis measurement. 
\end{theorem}

\begin{proof}
Because $\tr{\mathbb{F}_k B^{(m)}_k} = \pm 1$, we can write $\tr{\mathbb{F}_k B^{(m)}_k} = (-1)^{z^{(m)}_k}$ to convert our two outcomes from $\{-1,1\}$ to $\{0,1\}$. We can then let $\boldsymbol{z}^{(m)}=z^{(m)}_1 z^{(m)}_2 \dotsm z^{(m)}_n$, with $z^{(m)}_j \in \{0,1\}$, denote the bit string of length $n$ obtained as the outcome of the $m$-th measurement of our $n$ pairs of qubits in the test and copy states. In this notation, our estimator becomes
\begin{align}
     \E{\hat{\mathcal{C}}_{\ket{\psi}}(s)} &= 1 - \frac{1}{M}\sum_{m=1}^M \prod_{k \in s} \left(\frac{1+(-1)^{z^{(m)}_k}}{2} \right).
\end{align}
We can now show that this is an unbiased estimator. We obtain
\begin{align}
     \E{\hat{\mathcal{C}}_{\ket{\psi}}(s)} &=  \E{1- \frac{1}{M}\sum_{m=1}^M \prod_{k \in s} \left(\frac{1+(-1)^{z^{(m)}_k}}{2} \right)},\\
    &=  1- \frac{1}{M}\sum_{m=1}^M \E{\prod_{k \in s} \left(\frac{1+(-1)^{z^{(m)}_k}}{2} \right)},\\
    &= 1 - \frac{1}{M}\sum_{m=1}^M \sum_{\boldsymbol{z}^{(m)}} p(\boldsymbol{z}^{(m)}) \prod_{k\in s} \left(\frac{1+(-1)^{z^{(m)}_k}}{2} \right),\\
    &= 1 - \frac{1}{M}\sum_{m=1}^M \sum_{\boldsymbol{z}^{(m)}} \tr{\bigotimes_{k \in s} \frac{\mathbb{I}_k + (-1)^{z^{(m)}_k} \mathbb{F}_k}{2}\rho^{\otimes 2}} \underbrace{\prod_{k\in s} \left(\frac{1+(-1)^{z^{(m)}_k}}{2} \right)}_{\delta_{\boldsymbol{z}^{(m)},\bold{0}}},\\
    &= 1 - \frac{1}{M}\sum_{m=1}^M \tr{\bigotimes_{k \in s} \frac{\mathbb{I}_k + \mathbb{F}_k}{2} \rho^{\otimes 2}},\\
    &= 1 - \tr{\bigotimes_{k \in s} \frac{\mathbb{I}_k + \mathbb{F}_k}{2} \rho^{\otimes 2}},\\
     \E{\hat{\mathcal{C}}_{\ket{\psi}}(s)} &= 1- \frac{1}{2^{|s|}} \sum_{\alpha \in \mathcal{P}(s)} \tr{\rho_{\alpha}^2},
\end{align}
as desired.
\end{proof}
This proof points to a new, simpler interpretation of the CEs. We note that the product
\begin{align}
     \prod_{k=1}^{n} \left(\frac{1+(-1)^{z^{(m)}_k}}{2}\right),
\end{align}
is only nonzero if all $z_k^{(m)}=0$ (i.e. if the measurement round yields all triplet states). Thus, in a given measurement round, the two relevant outcomes are ``all triplet'' or ``at least one singlet,'' making each measurement round a Bernoulli trial. The probability of these two outcomes must sum to unity
\begin{align}
    1 &= \text{Prob}\left(\text{all triplets}\right) + \text{Prob}\left(\text{at least one singlet}\right).
\end{align}
This observation gives the following simple estimator of the CEs
\begin{align}
    \hat{C} &= \frac{\text{number of measurement rounds yielding at least one singlet}}{\text{total number of measurement rounds}},\label{eq:SimpleEstimator}\\
    \hat{C} &= 1 - \sum_{\ZC_0 (s)} p(\boldsymbol{z}),
\end{align}
where $\ZC_0 (s)$ is the set of bit-strings with zeroes on all indices in $s$. This recovers, and provides better intuition for, Prop. 1 in Ref. \cite{beckey2021computable}. 

\subsection{Unbiased estimation of $n$-tangle via Bell basis measurements}
    
\begin{theorem}
The quantity 
\begin{align}
    \hat{\tau}_{(n)} &= \frac{2^n}{M}\sum_{m=1}^M \prod_{k=1}^{n} \left(\frac{1-\tr{\mathbb{F}_k B^{(m)}_k}}{2} \right),
\end{align}
is an unbiased estimator of the $n$-tangle. That is,
\begin{align}
    \E{\hat{\tau}_{(n)}} &= \tau_{(n)},
\end{align}
where the expectation value is over all possible measurement outcomes. 
\end{theorem}
\begin{proof}
As above, let $\boldsymbol{z}^{(m)}=z^{(m)}_1 z^{(m)}_2 \dotsm z^{(m)}_n$, with $z^{(m)}_j \in \{0,1\}$ denote the bit string of length $n$ obtained as the outcome of the $m$-th measurement of our $n$ pairs of qubits in the test and copy states. Because $\tr{\mathbb{F}_k B^{(m)}_k} = \pm 1$, we can write $\tr{\mathbb{F}_k B^{(m)}_k} = (-1)^{z^{(m)}_k}$. Thus, our estimator becomes
\begin{align}
    \hat{\tau}_{(n)} &=  \frac{2^n}{M}\sum_{m=1}^M \prod_{k=1}^{n} \left(\frac{1-(-1)^{z^{(m)}_k}}{2} \right).
\end{align}
We can now show that this is an unbiased estimator of the $n$-tangle. We have
\begin{align}
    \E{\hat{\tau}_{(n)}} &=  \E{ \frac{2^n}{M}\sum_{m=1}^M \prod_{k=1}^{n} \left(\frac{1-(-1)^{z^{(m)}_k}}{2} \right)},\\
    &=  \frac{2^n}{M}\sum_{m=1}^M \E{\prod_{k=1}^{n} \left(\frac{1-(-1)^{z^{(m)}_k}}{2} \right)},\\
    &= \frac{2^n}{M}\sum_{m=1}^M \sum_{\boldsymbol{z}^{(m)}} p(\boldsymbol{z}^{(m)}) \prod_{k=1}^{n} \left(\frac{1-(-1)^{z^{(m)}_k}}{2} \right),\\
    &= \frac{2^n}{M}\sum_{m=1}^M \sum_{\boldsymbol{z}^{(m)}} \tr{\bigotimes_{k=1}^n \frac{\mathbb{I}_k - (-1)^{z^{(m)}_k} \mathbb{F}_k}{2}\rho^{\otimes 2}} \underbrace{\prod_{k=1}^{n} \left(\frac{1-(-1)^{z^{(m)}_k}}{2} \right)}_{\delta_{\boldsymbol{z}^{(m)},\bold{1}}},\\
    &= \frac{2^n}{M}\sum_{m=1}^M \tr{\bigotimes_{k=1}^n \frac{\mathbb{I}_k - \mathbb{F}_k}{2} \rho^{\otimes 2}},\\
    &=  2^n \tr{\bigotimes_{k=1}^n \frac{\sigma_{0_k}- \frac{1}{2}\left(\sigma_{0_k} + \sigma_{1_k} + \sigma_{2_k} + \sigma_{3_k}\right)}{2} \rho^{\otimes 2}},\\
    &= \frac{1}{2^n} \tr{\bigotimes_{k=1}^n (\sigma_{0_k}- \sigma_{1_k} - \sigma_{2_k} - \sigma_{3_k}) \rho^{\otimes 2}},\\
    &=\frac{1}{2^n} \left((S_{0\dots 0})^2 - \sum_{k=1}^n \sum_{i_k}^3 (S_{0\dots i_k \dots 0})^2 + \sum_{k,l=1}^n \sum_{i_k, i_l =1}^3 (S_{0\dots i_k \dots i_l \dots 0})^2 - \dotsm + (-1)^n \sum_{i_1, \dots, i_n}^3 (S_{i_1 \dots i_n})^2\right),\\
    &=S^2_{(n)},\\
    \E{\hat{\tau}_{(n)}} &=\tau_{(n)},
\end{align}
where the last three lines utilize the results of Ref. \cite{jaeger2003entanglement} in which the $n$-tangle is written in terms of the so-called $n$-qubit Stokes parameters defined as $S_{i_1,\dots,i_n}=\tr{\rho \sigma_{i_1}\otimes \dotsm \otimes \sigma_{i_n}}$ for $i_1,\dots,i_n \in \{ 0,1,2,3 \}$.
\end{proof}

As with the CEs, this proof points to a new, simpler interpretation of the $n$-tangle. We note that the product
\begin{align}
     \prod_{k=1}^{n} \left(\frac{1-(-1)^{z^{(m)}_k}}{2}\right),
\end{align}
is only nonzero if all $z_k^{(m)}=1$ (i.e. if the measurement round yields all singlet states). Thus, when one is interested in estimating the $n$-tangle, the two relevant outcomes are ``all singlet'' or ``at least one triplet.'' This observation gives the following simple estimator of the $n$-tangle
\begin{align}
    \hat{\tau} &= 2^n \cdot \frac{\text{number of measurement rounds yielding singlets}}{\text{total number of measurement rounds}},\\
    \hat{\tau} &= 2^n p(\boldsymbol{1}).
\end{align}
We note that this derivation recovers, and provides intuition for, Proposition 5 of Ref. \cite{beckey2021computable}.
\subsection{Number of measurements required for $\epsilon$-close estimations}
\subsubsection{Analytical method}
\begin{proposition}
    Let $\epsilon, \delta > 0$ and $M=\Theta \left({\frac{\log{1/\delta}}{\epsilon^2}}\right)$. Further, let $\theta \in \{\mathcal{C}_{\ket{\psi}}(s), \tau_{(n)}\}$ and let $\hat{\theta}$ denote the corresponding estimator for $\theta$. Then we have 
    \begin{align}
        \left|  \hat{\theta} - \theta \right| < \epsilon,
    \end{align}
    with probability at least $1-\delta$.
\end{proposition}
\begin{proof}
From Hoeffding's inequality (Fact \ref{fact:hoeffding} above), we can write
\begin{align}
    \text{Prob}\left(\left| \hat{\theta} - \theta \right| < \epsilon \right) \geq 1 - 2\exp{\left(-\frac{2 M^2 \epsilon^2}{\sum_{m=1}^M (1-0)^2}\right)},
\end{align}
because our quantities satisfy $0 \leq \theta \leq 1$. Thus, Hoeffding's inequality tells us that to achieve an $\epsilon$-close approximation of the elements of $\theta$, with probability at least $1-\delta$, one would require at least
\begin{align}
\delta &= 2\exp{\left(\frac{-2M^2 \epsilon^2}{\sum_{m=1}^M (1-0)^2} \right)},\\
\delta &= 2\exp{\left(-2M \epsilon^2 \right)},\\
\implies M &= \frac{\log{2 / \delta}}{2\epsilon^2}.
\end{align}
measurements. The constants are ignored in big-$\Theta$ notation, yielding the desired result that
\begin{align}
    M=\Theta\left(\frac{\log{1/\delta}}{\epsilon^2}\right)
\end{align}
measurements are needed to obtain an $\epsilon$-close estimate of $\theta$ with high probability. 
\end{proof}
\subsubsection{Numerical method}
While the above method, based on Hoeffding's inequality, is nice for deriving analytical scaling, it does not take into account the underlying distribution and is likely not tight as a result. As discussed in Supp. Sec.~\ref{sec:CE-from-bell}, the CE can be estimated by simply computing the probability of obtaining all triplet states on the systems measured. As such, we can regard the Bell basis measurement of all qubit pairs as a Bernoulli trial in which $X=0$ corresponds to measuring all pairs in a triplet state and $X=1$ to measuring at least one pair in the singlet state. Then, according to \eqref{eq:SimpleEstimator} estimating the CE is equivalent to performing a MLE of the binomial probability $\theta=P(X=1)$. Since the number of times $k$ that $X=1$ is obtained in $M$ measurement rounds is a random variable with probability mass function given by \eqref{eq:Binomial_pmf}, we can compute the average upper and lower bounds of the $100\times (1-\delta)\%$ confidence interval as
 \begin{align}
&\langle \theta_U\rangle(M) =\sum_{k=0}^{M}P(k)\theta_U(k)=\sum_{k=0}^{M} {{M}\choose{k}}\theta^k(1-\theta)^{M-k} \theta_U(k),\\
&\langle \theta_L\rangle(M) =\sum_{k=0}^{M}P(k)\theta_L(k)=\sum_{k=0}^{M} {{M}\choose{k}}\theta^k(1-\theta)^{M-k} \theta_L(k),
\end{align}
where each value of $\theta_U(k), \theta_L(k)$ is found by solving Eqs.~\ref{eq:CPU},\ref{eq:CPL}. Taking the average of the upper and lower bounds, for each number of measurements $M$ we can say that
\begin{equation}
\left|  \hat{\theta} - \theta \right| < \frac{\langle \theta_U\rangle(M)-\langle \theta_L\rangle(M)}{2}:= \epsilon
\end{equation}
with probability $1-\delta$. 

\subsection{Concentratable Entanglement of Mixed States}
\subsubsection{Lower bound on the mixed state bipartite concurrence}
Recall that the concurrence of a pure bipartite quantum state, $\ket{\psi} \in \HC_A \otimes \HC_B$, is given as
\begin{align}
    c_2 (\ket{\psi}) &= \sqrt{2(1-\tr{\rho_A^2})}.
\end{align}
The standard convex-roof extension can then be used make this measure well defined for mixed state inputs \cite{bennett1996mixed,uhlmann2010roofs}
\begin{align}
    c_2 (\rho) &= \inf \sum_i p_i c_2 (\ket{\psi_i}),
\end{align}
where $\rho=\sum_i p_i \ket{\psi_i}\bra{\psi_i}$ and $\sum_i p_i =1.$ To avoid having to do any optimization, the authors of Ref.~\cite{mintert2007observable} introduce an observable lower bound on the bipartite concurrence given as
\begin{align}
    [c_2(\rho)]^2 \geq 2\tr{\rho^2} -\tr{\rho_A^2} - \tr{\rho_B^2}.
\end{align}
When $\tr{\rho^2}=1$, $\tr{\rho_A^2}=\tr{\rho_B^2}$ via the Schmidt decomposition and we recover exactly the pure state bipartite concurrence squared $2(1-\tr{\rho_A^2})$. 
\subsubsection{Lower bound on mixed state CEs}
With the above bipartite bound in mind, we can construct lower bounds on $\CC_{\rho}(s)$ using the relationship between CEs and the bipartite concurrences  $c_{\alpha}(\ket{\psi}) := \sqrt{2(1-\tr{\rho_{\alpha}^2})}$ \cite{wootters2001entanglement}. As stated in the main text, any CE can be expressed in terms of bipartite concurrences as
\begin{align}
    \CC_{\ket{\psi}} (s) &= \frac{1}{2^{|s|+1}} \sum_{\alpha} c_{\alpha}^2(\ket{\psi}),\\
    &= \frac{1}{2^{|s|+1}} \sum_{\alpha} 2(1-\tr{\rho_{\alpha}^2}),\\
    &= \frac{1}{2^{|s|}} \sum_{\alpha} (1-\tr{\rho_{\alpha}^2}),\\
    \CC_{\ket{\psi}} (s) &= 1 -\frac{1}{2^{|s|}}\sum_{\alpha} \tr{\rho_{\alpha}^2}
\end{align}
Then, because we can use the known lower bound for each bipartite concurrence in the sum, we can construct a lower bound for any CE of interest. This is a generalization of the method used in Ref. \cite{aolita2008scalable} in which the authors derive a lower bound on the mixed state multipartite concurrence. 

For example, the lower bound on $\CC_{\rho}(\SC)$ takes the form
\begin{align}
    \CC_{\rho}^{\ell}(\SC) = \frac{1}{2^n} + (1-\frac{1}{2^n})\tr{\rho^2} - \frac{1}{2^n}\sum_{\alpha \in \PC (\SC)} \tr{\rho_{\alpha}^2}.
\end{align}
Because each term in this expression can be directly estimated from Bell basis measurement data, it allows one to quantify mixed state entanglement in the same framework developed above for pure state entanglement. That is, using the estimators for $n$-qubit purity (Eq. \ref{eq:subsytem-purity} and uniform average subsystem purity (by slight modification of the estimator from Thm.~\ref{thm:CE-from-Bell}),  one can estimate the lower bound from Bell basis measurement data via
\begin{align}
    \hat{\CC}_{\rho}^{\ell}(\SC)  = \frac{1}{2^n}  +(1-\frac{1}{2^n})\cdot \frac{1}{M} \sum_{m=1}^M \prod_{k=1}^n \tr{\mathbb{F}_k B_{k}^{(m)}} - \frac{1}{M}\sum_{m=1}^M \prod_{k=1}^n  \frac{1+\tr{\mathbb{F}_k B_k^{(m)}}}{2}.
\end{align}
Similar bounds can be constructed for all $s\subseteq \SC$. These bounds, as well as all of the other measures discussed above, can all be estimated from the same Bell basis measurement data.

\section{Simulation Details}\label{sec:simulations}
\subsection{Ordering of quantum states under CE}
Every entanglement measure puts an ordering on quantum states. By definition, fully separable states must evaluate to zero. However, it is not always clear what states extremize a given entanglement measure. While we do not answer this question in full generality, we do provide analytical formulas for the CEs of W, GHZ, and Line states -- all of which are plotted in Fig. 4 of the main text. We also note that the question of what states extremize the CE has recently been investigated in Ref. \cite{schatzki2022hierarchy}. 
\subsubsection{Analytical formula for $\CC_{|{GHZ}\rangle}(\SC)$}
We begin with the simplest state, GHZ. Recall that an $n$-qubit GHZ state is defined as 
\begin{align}
    \ket{\text{GHZ}_n} &=\frac{1}{\sqrt{2}} (\ket{\vec{0}}+\ket{\vec{1}}).
\end{align}
As one can show, tracing out one or more qubits yields a reduced state with purity $1/2$. Thus, all terms in the CE (except the empty set and the full set) correspond to a reduced state purities of $1/2$. That is, we can write
\begin{align}
\CC_{\ket{\text{GHZ}_n}}(\SC)&=1-\frac{1}{2^n}\sum_{\alpha\in \PC(\SC)} {\rm Tr} \rho_\alpha^2,\\
&= 1 - \frac{1}{2^n} \left(2 + \frac{1}{2} (2^n -2 )\right),\\
\CC_{\ket{\text{GHZ}_n}}(\SC) &= \frac{1}{2} - \frac{1}{2^n}.
\end{align}
We note that this formula was found numerically in Ref. \cite{foulds2021controlled} and analytically in Ref. \cite{cullen2022calculating}. 
\subsubsection{Analytical formula for $\CC_{\ket{W}}(\SC)$}\label{sec:W-state-CE}

Next, we consider W states. Recall they are defined as the equal superposition of all states with labels that have hamming weight 1. That is
\begin{align}
    \ket{W_n} & = \frac{1}{\sqrt{n}}(\ket{10\dotsm 0} + \ket{010\dotsm 0} + \dotsm + \ket{0\dotsm 01}).
\end{align}
For our purposes, it is instructive to note that W states can be generated recursively as
\begin{align}
    \ket{W_2}&= \frac{1}{\sqrt{2}} (\ket{10} +\ket{01}),\\
    \ket{W_n}&= \frac{\sqrt{n-1}}{\sqrt{n}} \ket{W_{n-1}}\otimes \ket{0} + \frac{1}{\sqrt{n}}\ket{\boldsymbol{0}_{n-1}} \otimes \ket{1}, 
\end{align}
where $\boldsymbol{0}_{n-1}$ denotes the all-zero bit string of length $n-1$. If we then let A be a subspace of $(\mathbb{C}^2)^{\otimes n}$ with dimension $d$, and define $\log{d}:=j$, then, for all $0\leq j \leq n-1$, one can show that 
    \begin{align}
        \text{tr}_{\text{A}}{W_n} &= \frac{n-j}{n} \ket{W_{n-j}}\bra{W_{n-j}} + \frac{j}{n}\ket{\vec{0}_{n-1}}\bra{\vec{0}_{n-1}}.
    \end{align}

It follows that
\begin{align}
    \tr{( \text{tr}_{\text{A}}{W_n} )^2} &= \frac{(n-i)^2+j^2}{n^2}.
\end{align}
Having expressed the purity of all reduced density matrices in terms of the number of qubits and the dimension of the subspace that has been traced out, we can find a closed form expression for the W-state CE via 
\begin{align}
\CC_{\ket{\text{W}_n}}(\SC)&=1-\frac{1}{2^n}\sum_{\alpha\in \PC(\SC)} {\rm Tr} \rho_\alpha^2,\\
&= 1 - \frac{1}{2^n}\left(2+\sum_{i=1}^{n-1} {n \choose i} \frac{(n-i)^2+i^2}{n^2}\right),\\
\CC_{\ket{\text{W}_n}}(\SC) &= \frac{1}{2} -\frac{1}{2n},
\end{align}
where the sum was evaluated and simplified using Mathematica. This proves the empirically derived formulas in Refs. \cite{foulds2021controlled,beckey2021computable}.

\subsubsection{Analytical formula for $\CC_{\ket{L_n}}(\SC)$}
Line states, which we denote $\ket{L_n}$, are a special case of a broader class of states known as graph states~\cite{hein2004multiparty,cullen2022calculating}. They don't admit as simple a representation as W or GHZ states, but, as we will see, they're are far more entangled than W or GHZ states. 
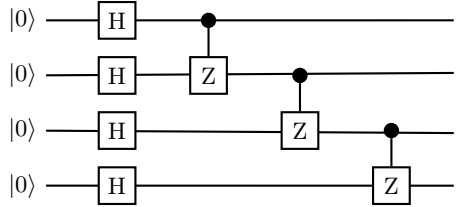
\begin{figure}[h!]
    \centering
    \begin{tikzpicture}[thick]
        \tikzset{operator/.style = {draw,fill=white,minimum size=1.5em},
        phase2/.style = {draw,fill=black,shape=circle,minimum size=5pt,inner sep=0pt}, 
        ,}
        \matrix[row sep=0.2cm, column sep=0.7cm] (circuit) {\node (q1) {$\ket{0}$};
        &\node[operator] (H1) {H};
        &\node[phase2] (P1) {};
        \\ \node (q2) {$\ket{0}$};
        &\node[operator] (H2) {H};
        &\node[operator] (Z2) {Z};
        &\node[phase2] (P2) {};
        \\ \node (q3) {$\ket{0}$}; 
        &\node[operator] (H3) {H};
        & &[-1.2cm] \node[operator] (Z3) {Z};
        &\node[phase2] (P3) {};
        \\ \node (q4) {$\ket{0}$}; 
        &\node[operator] (H4) {H};
        & &[-1.2cm] &[-1.2cm] \node[operator] (Z4) {Z};
        \\};
        \begin{pgfonlayer}{background} \draw[thick] (q1) -- (3.3, 1.1) (q2) -- (3.3, 0.4) (q3) -- (3.3, -0.4) (q4) -- (3.3, -1.1); \draw[thick] (P1) -- (Z2) (P2) -- (Z3) (P3) -- (Z4);
        \end{pgfonlayer}
    \end{tikzpicture}
\caption{\textbf{Line State Preparation.} Circuit diagram used to prepare a 4-qubit line state.}
\label{fig:line-state}
\end{figure}

Sup. Fig.~\ref{fig:line-state} shows the circuit used to prepare a $4$-qubit line state. The general $n$-qubit circuit follows the same pattern. Simply start in $\ket{+}$ state and then applies $CZ$ gates between all nearest neighbors. Note that one does not connect the $n$-th qubit to the first (this would be a different type of graph state called a ring state). One finds the following remarkable formula for the Line state CE
\begin{align}
    \CC_{\ket{L_n}}(\SC) &= 1- \frac{\text{Fibonacci}[n+2]}{2^n},
\end{align}
where $\text{Fibonacci}[n+2]$ denotes the $(n+2)$-th term in the Fibonacci sequence generated recursively via
\begin{align}
    \text{Fibonacci}[1]&=1, \\
    \text{Fibonacci}[2]&=1, \\
    \text{Fibonacci}[n]&= \text{Fibonacci}[n-1]+\text{Fibonacci}[n-2], \quad \forall \quad n \geq 3.
\end{align}
This unexpected formula was found numerically. The methods used for the $W$ state do not seem applicable to the Line state. However, a proof is likely possible using the methods recently introduced in Ref. \cite{cullen2022calculating}.

\subsection{Realistic quantum gates with Rydberg atoms}
\begin{figure}[h!]
\centering
\includegraphics[width=0.5\linewidth]{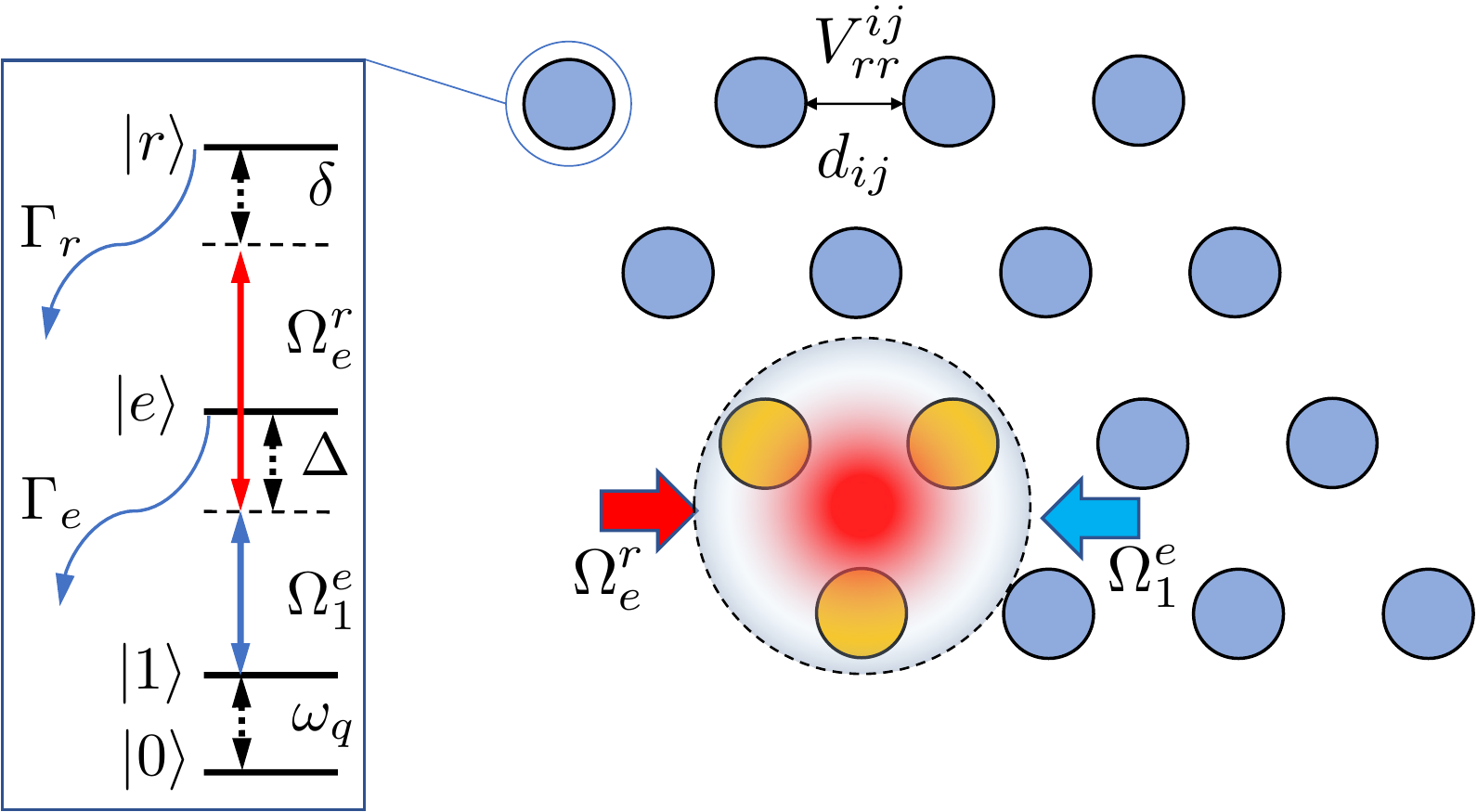}
\caption{\textbf{Array of trapped neutral atoms.}  The inset shows the energy levels and parameters considered in the model.}
\label{schemegates}
\end{figure}
In this section we outline the main ideas of Ref. \cite{pelegri2022high}, which we have used to model a realistic measurement of the CE in a Rydberg system using quantum gates. In Fig. \ref{schemegates} we sketch the general physical system that we have in mind. Neutral alkali atoms are trapped with optical tweezers in a lattice of arbitrary geometry, and the qubits are encoded in long-lived hyperfine ground states. Entangling operations are facilitated by coupling of the logical state $\ket{1}$ to a highly excited Rydberg state $\ket{r}$ via a two-photon excitation scheme through a far-detuned intermediate state $\ket{e}$, which is resolved into its hyperfine components $\ket{f_e,m_{f_e}}$. The Hamiltonian describing this excitation process can be written as

\begin{align}
\frac{\mathcal{H}}{\hbar} &=\sum_{f_e,m_{f_e}} \frac{1}{2} \left(\Omega_1^{ f_e,m_{f_e}}\ket{1}\bra{f_e,m_{f_e}} +{\Omega_1^{ f_e,m_{f_e}}}^*\ket{f_e,m_{f_e}}\bra{1}\right) \nonumber\\
&-\sum_{f_e,m_{f_e}}\Delta_{f_e,m_{f_e}} \ket {f_e,m_{f_e}}\bra{ f_e,m_{f_e}}\nonumber\\
&+\sum_{f_e,m_{f_e}} \frac{1}{2}\left(\Omega_{f_e,m_{f_e}}^{r} \ket{f_e,m_{f_e}} \bra{r}+{\Omega_{f_e,m_{f_e}}^{r}}^*\ket{r}\bra{f_e,m_{f_e}}\right)\nonumber\\
&-\delta\ket{r}\bra{r},
\end{align}
where $\Omega_1^{ f_e,m_{f_e}}$ and $\Omega_{f_e,m_{f_e}}^{r}$ are the Rabi frequencies of the drives from $\ket{1}$ to $\ket{f_e,m_{f_e}}$ and from $\ket{f_e,m_{f_e}}$ to $\ket{r}$,  $\Delta_{f_e,m_{f_e}}=\Delta-E(f_e,m_{f_e})$ represent the intermediate-state detunings composed by the laser detuning $\Delta$ and the hyperfine splittings $E(f_e,m_{fe})$, and $\delta$ is the total two-photon detuning. The excitation process suffers from losses due to the finite linewidths $\Gamma_e$ and $\Gamma_r$ of the states $\ket{e}$ and $\ket{r}$. We describe these scattering process by introducing effective non-hermitian terms in the Hamiltonian given by
\begin{equation}
\mathcal{H}' =-i\hbar\displaystyle\sum_{f_e,m_{f_e}} \Big[ \Gamma_e/2\ket {f_e,m_{f_e}}\bra{ f_e,m_{f_e}}\Big]+\Gamma_r/2\vert r \rangle \langle r\vert.
\end{equation}
The Rydberg states experience dipole-induced pairwise interaction described by the Hamiltonian
\begin{equation}
\mathcal{H}_\mathrm{dd}=\displaystyle\sum_{j<i}\hbar V_{rr}^{ij}\vert r_ir_j\rangle\langle r_ir_j\vert,
\end{equation}
where the strength $V_{rr}^{ij}$ depends on the separation $d_{ij}$ between the atoms $i$ and $j$ and their orientation with respect to the quantization axis. Under time evolution with the total Hamiltonian $\mathcal{H}_{\mathrm{tot}}=\mathcal{H}+\mathcal{H}'+\mathcal{H}_{\mathrm{dd}}$, we seek to apply global laser pulses to an ensemble of $k+1$ atoms which realize a multiply controlled phase gate $C^{k}Z$ described by the unitary transformation 
\begin{equation}
U_{C^{k}Z}=2\left(\otimes_{k+1}\ket{0} \otimes_{k+1}\bra{0}\right)-I.
\end{equation}
As described in \cite{pelegri2022high}, this is achieved by working in the fully Rydberg blockaded regime and performing Adiabatic Rapid Passage from the ground state $\ket{1}$ to the manifold of states with a single Rydberg excitation. For the purposes of measuring the CE, we are interested in realising a $CZ$ gate (1 control atom) and a $CCZ$ gate (2 control atoms), which can be converted respectively into CNOT and Toffoli gates by application of additional Hadamard gates to the target atoms.  In order to maximise the gate fidelity, the Rydberg interactions $V_{rr}^{ij}$ need to be as large as possible - i.e., the atoms involved in the gate need to be as close as possible- and the pulses should be designed to minimise the losses due to photon scattering from $\ket{e}$ and $\ket{r}$. In a concrete setting with Cs atoms and realistic parameters, after pulse optimisation we obtain the following effective matrices for the $CZ$ and $CCZ$ gates
\begin{align}
U_{CZ}&=\ket{00}\bra{00}+0.9990e^{i0.9906\pi}(\ket{01}\bra{01}+\ket{10}\bra{10})+0.9986e^{i1.000\pi}\ket{11}\bra{11}\label{matCZ},\\
U_{CCZ}&=\ket{000}\bra{000}+0.9981e^{i0.9845\pi}(\ket{001}\bra{001}+\ket{010}\bra{010}+\ket{100}\bra{100})\nonumber\\
&+0.9973e^{i0.9934\pi}(\ket{110}\bra{110}+\ket{101}\bra{101}+\ket{011}\bra{011})+0.9963e^{i0.9911\pi}\ket{111}\bra{111}\label{matCCZ}.
\end{align}
Note that these matrices are not unitary due to the loss of population caused by the scattering. All the results shown in the main text have been obtained by simulating the c-SWAP test and Bell measurement quantum circuits using these effective gate matrices and assuming perfect single-qubit gates.

\end{document}